\documentclass{lmcs}
\pdfoutput=1
\usepackage{lastpage}
\lmcsdoi{22}{2}{18}
\lmcsheading{}{\pageref{LastPage}}{}{}%
{Aug.~29,~2025}{May~20,~2026}{}


\usepackage[utf8]{inputenc}
\usepackage{amsmath}
\usepackage{amssymb}
 \usepackage{amsthm} 
\usepackage{graphicx}

\usepackage{tikz}
\usepackage{tikz-qtree}
\usetikzlibrary{shapes.geometric, arrows}
\usetikzlibrary{decorations.markings}
\tikzset{negated/.style={
		decoration={markings,
			mark= at position 0.5 with {
				\node[transform shape] (tempnode) {$\backslash$};
			}
		},
		postaction={decorate}
	}
}
\usetikzlibrary{hobby} 

\usepackage{virginialake}

\usepackage[disable]{todonotes}

 \usepackage{hyperref}

\usepackage[capitalise]{cleveref}
\crefname{propC}{Proposition}{Propositions}
\usepackage[inline]{enumitem}

\crefformat{enumi}{(#2#1#3)}

\newcommand{\dt}{{DT}}
\newcommand{\exdt}{e\dt}
\newcommand{\ndt}{{N}\dt}
\newcommand{\exndt}{e\ndt}

\newcommand{\Bool}{\mathrm{Bool}}
\newcommand{\boolexdt}{$\Bool$(\exdt)}
\newcommand{\boolexndt}{$\Bool$(\exndt)}

\newcommand{\szel}{Szelepcs\'{e}nyi}

\newcommand{\poly}{\mathrm{poly}}

\newcommand{\emphasis}[1]{\textbf{#1}}
\renewcommand{\epsilon}{\varepsilon}
\renewcommand{\phi}{\varphi}

\newcommand{\dfn}{:=}
\newcommand{\bnf}{::=}
\renewcommand{\vec}[1]{\mathbf{#1}}

\newcommand{\IH}{\mathit{IH}}

\newcommand{\storageone}{}
\newcommand{\storagetwo}{}

\newcommand{\orange}[1]{{\color{orange} #1}}
\newcommand{\gray}[1]{{\color{gray} #1}}
\newcommand{\blue}[1]{{\color{blue} #1}}

\newcommand{\red}[1]{{\color{red} #1}}

\newcommand{\green}[1]{{\color{teal} #1}}

\newtheorem{theorem}[thm]{Theorem}
\newtheorem{proposition}[thm]{Proposition}
\newtheorem{lemma}[thm]{Lemma}
\newtheorem{corollary}[thm]{Corollary}

\theoremstyle{definition}
\newtheorem{definition}[thm]{Definition}
\newtheorem{remark}[thm]{Remark}
\newtheorem{example}[thm]{Example}

\newcommand{\co}{\mathit{co}}
\newcommand{\N}{\mathbf{N}}
\newcommand{\Logspace}{\mathbf{L}}
\newcommand{\NL}{\N\Logspace}
\newcommand{\coNL}{\co\NL}

\newcommand{\Ptime}{\mathbf{P}}
\newcommand{\NP}{\mathbf{NP}}
\newcommand{\coNP}{\co\NP}
\newcommand{\ALOGTIME}{\mathbf{ALOGTIME}}

\newcommand{\bool}{\{0,1\}}

\newcommand{\unf}[2]{{\mathrm{Unf}_{#2}(#1)}}

\newcommand{\exndtsize}[1]{\# #1}

\newcommand{\sat}[1]{\vDash_{#1}}
\newcommand{\notsat}[1]{\nvDash_{#1}}

\newcommand{\simulates}[1]{\gtrsim_{#1}}

\newcommand{\cnot}{\neg}

\newcommand{\cand}{\wedge}
\newcommand{\cimp}{\supset}

\newcommand{\dec}[3]{#1#2#3}
\newcommand{\posdec}[3]{#1#2(#1\lor #3)}

\newcommand{\decider}[4]{\dec{#2}{#3^{#1}}{#4}}

\newcommand{\extiff}{\leftrightarrow}

\newcommand{\ee}[2]{e_{#1#2}}

\newcommand{\thr}[2]{t^{#1}_{#2}}
\newcommand{\thresh}[2]{\thr{#1}{#2}}

\newcommand{\thrextaxs}[2]{\mathcal{T}^{#1}_{#2}}
\newcommand{\dextaxs}[2]{\mathcal D^{#1}_{#2}}

\newcommand{\E}{\mathcal{E}}
\newcommand{\BB}{\mathcal{B}}

\newcommand{\seqar}{\mbox{\Large $\, \rightarrow \, $}}
\newcommand{\dseqar}{\extiff}
\newcommand{\lk}{\mathsf{LK}}

\newcommand{\ext}{\mathsf{e}}
\newcommand{\ldt}{\mathsf{LDT}}
\newcommand{\lndt}{\mathsf{LNDT}}

\newcommand{\eldt}{\ext\ldt}
\newcommand{\elndt}{\ext\lndt}

\newcommand{\booleldt}{\lk(\eldt)}

\newcommand{\posboolelndt}{\lk^+ ( \elndt)}

\newcommand{\DM}{\mathrm{DM}}

\newcommand{\bl}{\mathsf{NB}}
\newcommand{\blus}{\bl^\DM}

\newcommand{\dbl}{\mathsf{DB}}

\newcommand{\dm}[1]{#1^\DM}
\newcommand{\depth}[1]{\mathrm{dp}(#1)}

\newcommand{\pos}[1]{#1^{+}}

\newcommand{\lefrul}[1]{#1\text{-}l}
\newcommand{\rigrul}[1]{#1\text{-}r}

\newcommand{\infrule}{\mathsf{r}}

\newcommand{\id}{\mathsf{id}}
\newcommand{\exch}{\mathsf{e}}
\newcommand{\cntr}{\mathsf{c}}
\newcommand{\wk}{\mathsf{w}}
\newcommand{\cut}{\mathsf{cut}}

\newcommand{\List}[3]{\vec{#2}^{#3}}

\newcommand{\nndt}{(N)DT}
\newcommand{\exnndt}{e(N)DT}
\newcommand{\lnndt}{\mathsf{L(N)DT}}
\newcommand{\elnndt}{\mathsf{eL(N)DT}}
\newcommand{\boolexnndt}{$\Bool$(\exnndt)}

\newcommand{\decext}[3]{d_{#1}^{#2,#3}}
\newcommand{\decextaxs}[1]{\mathcal D^{#1}}
\newcommand{\threshextaxs}[1]{\mathcal T^{\vec B}}

\newcommand{\B}{{\vec B}}
\newcommand{\Bj}[1]{\B^{#1}}

\begin{document}

\title[Prover-Adversary games for (non-deterministic) branching programs]{Prover-Adversary games for systems over (non-deterministic) branching programs}\thanks{The alphabetically first author has been supported by a UKRI Future Leaders Fellowship, \emph{Structure vs Invariants in Proofs}, project number
MR/S035540/1.}
\author[A.~Das]{Anupam Das\lmcsorcid{0000-0002-0142-3676}}
\author[A.~Delkos]{Avgerinos Delkos}


\address{University of Birmingham, UK}	
\email{a.das@bham.ac.uk, axd1010@alumni.bham.ac.uk} 

\begin{abstract}
 We introduce Pudl\'ak-Buss style Prover-Adversary games  to characterise proof systems reasoning over deterministic branching programs (BPs) and non-deterministic branching programs (NBPs).
 Our starting points are the proof systems $\eldt$ and $\elndt$, for BPs and NBPs respectively, previously introduced by Buss, Das and Knop.
 We prove polynomial equivalences between these proof systems and the corresponding games we introduce.
   This crucially requires access to a form of \emph{negation} of branching programs which, for NBPs, requires us to formalise a non-uniform version of the Immerman-\szel\ theorem that $\coNL=\NL$.
   Thanks to the techniques developed, we further obtain a proof complexity theoretic version of Immerman-\szel, showing that $\elndt$ is polynomially equivalent to systems over boundedly alternating branching programs.
\end{abstract}

\maketitle

\section{Introduction}

\emph{Proof complexity} investigates the size of proofs, in terms of that of their conclusions.
Originally motivated by the Cook-Levin theorem that Boolean satisfiability is $\NP$-complete, finding general superpolynomial lower bounds on proof size for Boolean logic directly implies $\Ptime\neq \NP$,
what is now known as \emph{Cook's programme}.
Systems of interest in proof complexity are typically parametrised by a complexity class of interest, whose nonuniform counterpart comprises the objects of reasoning in the associated proof system.
For instance Frege systems reason on formulas, the nonuniform counterpart of $\ALOGTIME$ \cite{Buss87:BFV-in-ALOGTIME}. 
For $\Ptime$ the corresponding system is \emph{extended} Frege, employing `Tseitin extension' to represent the dag structure of circuits.

Recently Buss, Das and Knop proposed a proof complexity theory of $\Logspace $ and $\NL$ via systems, $\eldt $ and $\elndt$, reasoning about deterministic branching programs (BPs) and non-deterministic branching programs (NBPs) respectively \cite{DBLP:conf/csl/BussDasKnop}.
In earlier work \cite{Das-Del,DasDel25:pos-bps-journal} the current authors studied the `positive' fragment of $\elndt$, reasoning about \emph{monotone} branching programs, and showed a polynomial simulation over positive sequents, inspired by previous developments for the sequent calculus \cite{DBLP:journals/jcss/AtseriasGP02}.
In all these works (and often in proof complexity),
reasoning about families of propositional proofs can be cumbersome and notationally complex. 
This issue is made worse when the objects of reasoning have underlying dag structures, represented using extension where indexing must be carefully controlled to maintain well-foundedness, yielding further technicalities for handling program equivalence/simulation. 

A promising solution to these issues is the program of \emph{bounded arithmetic} \cite{buss1985bounded,krajicek1995bounded,cooknguyen2010logical,krajivcek2019proof}.
Here (very) weak theories of arithmetic essentially serve as uniform counterparts of propositional proof systems.
However in this work we follow a different branch of the proof complexity literature: \emph{Prover-Adversary games}, à la Pudl\'ak-Buss \cite{PudBuss}.
Roughly speaking, these are two-player games where Prover asks queries (typically the objects of reasoning) and Adversary assigns each query a Boolean value; Prover wins if they can force Adversary into a `simple contradiction' (typically one `easily' witnessed by polynomial-size proofs).
In this way proofs are recast as \emph{strategies}, with the logarithm of proof size proportional to the depth of a strategy.
Indeed, we can see such `game systems' as canonical `balanced tree-like' versions of their corresponding inference systems, cf.~\cite{krajivcek1994lower}.

    \begin{figure}[t]
        \centering
     \begin{tikzpicture}[scale=1.68]
     
     \foreach \pos/\name/\disp in {
  {(-2,3)/2/{$\bl$}},
  {(1.8,1.48)/3/{\small$\posboolelndt$}}, 
  {(-0.88,1.92)/5/{$\blus$}},
  {(2,2.8)/6/{\small $\elndt$}},
  {(1.68,3.48)/10/{\tiny$\to  B$}},
  {(-2.05,3.53) /10/{\tiny$B\mapsto 0$}},
  {(-1,2.12)  /10/{\tiny$B\mapsto 0$}},
    {(1.68,3.98)/12/{\tiny proof}},
      {(-2,4.08)/13/{\tiny strategy}},
      {(-1,2.72)/14/{\tiny De Morgan}},
        {(-1,2.58)/15/{\tiny strategy }},
          {(1.08,1.8)/16/{\tiny$\thresh{\List{}{B}{}}{k}\hspace{-2pt}\to B, \thresh{\List{}{B}{}}{k+1}$}},
           {(1.08,2.48)/17/{\tiny proof for}},
             {(1.08,2.34)/17/{\tiny each $k$}},
               {(1.08,2.2)/17/{\tiny }}}
\node[minimum size=20pt,inner sep=0pt] (\name) at \pos {\disp};
    
    
    \draw [->][thick,olive][out=268, in=188](2) to (5);
         \draw [->][thick,orange]
    (6) [out=128, in=68] to  (2);
  
    \draw [->][ultra thick,cyan](5) to  [out=-8, in=188](3) ;
    \draw [->][thick](3) to (6);
    \draw [-][thick,blue](1.68,3.58) to (2.18,4.18);
    \draw [-][thick,violet](1.68,3.58) to (1.18,4.18);  
     \draw [-][thick,violet]
    (1.18,4.18) [out=28, in=128] to  (2.18,4.18);
  \draw [-][thick,blue](-2,3.58) to (-1.5,4.18);
    \draw [-][thick,violet](-2,3.58) to (-2.5,4.18);  
     \draw [-][thick,violet]
    (-2.5,4.18) ..controls(-1.88,4.48) and (-1.7,4.05).. (-1.5,4.18);
   \draw [-][thick,blue](-1,2.18) to (-.5,2.78);
    \draw [-][thick,violet](-1,2.18) to (-1.5,2.78);  
     \draw [-][thick,violet]
    (-1.5,2.78) ..controls(-.88,3.08) and (-.7,2.65).. (-.5,2.78);
    
  
   \draw [-][thick,blue](1.08,1.88) to (1.58,2.48);
    \draw [-][thick,violet](1.08,1.88) to (.58,2.48);  
     \draw [-][thick,violet]
    (.58,2.48) [out=28, in=128] to  (1.58,2.48);
     \end{tikzpicture}

        \caption{High-level structure of the polynomial simulations for the non-deterministic setting.
        The right side displays inference systems and proofs therein, while the left side displays game systems and strategies therein.
        The bold cyan arrow is the most technical, incorporating our formalisation of Immerman-\szel.
}
        \label{High level view}
        \label{structure-of-paper}
    \end{figure}

\subsection*{Contributions}
In this work we develop games $\dbl$ and $\bl$ for reasoning about BPs and NBPs, respectively, and prove their correspondence to $\eldt$ and $\elndt$, respectively.
Classically, the equivalence of tree-like and dag-like versions of a system typically requires closure of the language under Boolean combinations (in particular negation).
Such bootstrapping is readily formalisable in $\eldt$, with relatively simple constructions of Boolean combinations of (deterministic) BPs (see, e.g., \cite{Weg00:bps-and-bdds}).
In fact this idea of using BP-based games for $\Logspace$ proof complexity was already proposed by Cook \cite{Cook01slides} (but, as far as we know, never published).
However the problem becomes more complex for $\elndt$, which comprises the main focus of this work.
Here the idea is to compute the \emph{negation} of an NBP as an NBP by formalising a partial non-uniform version of the Immerman-\szel\ theorem that $\coNL = \NL$ \cite{immerman1988nondeterministic,szelep1988method}.
A novelty of our formalisation is that we are able to simplify the inductive counting argument: we do not encode counting at the level of NBPs themselves, but rather at the level of the \emph{proof}, relying on constructions of (positive) counting programs from previous work \cite{Das-Del,DasDel25:pos-bps-journal}.  
Thus the programs we construct are \emph{partial} negations, working relative to a fixed number of true inputs.

Let us point out a particular design choice at this point.
In order to more easily reason about proof complexity via games, we simply close the queries of our games under (explicit) Boolean combinations.
Of course there is no free lunch here: the aforementioned formalisation of Immerman-\szel\ is duly carried out in the proof system $\elndt$.

As an application of the machinery built up in this work, we develop in \cref{sec:prf-comp-immszel} a bona fide proof complexity theoretic version of the Immerman-\szel\ theorem. 
Namely we show that $\elndt$ polynomially simulates a system, $\mathsf{eL\exists\forall DT}$, reasoning over twice-alternating branching programs ($\exists\forall$BPs). 
Here we exploit the previously constructed partial negations of NBPs to reduce any $\exists\forall$BP to an equivalent NBP, relative to a fixed number of true inputs. 
Again, the general polynomial simulation is recovered by using a counting argument at the proof level, exhausting every possible number of true inputs.

\subsection*{Structure}
In \Cref{sec:prelims-elndt} we recall the systems $\eldt$ and $\elndt$ from \cite{DBLP:conf/csl/BussDasKnop}, and in \Cref{sec:games} we introduce our corresponding games $\dbl$ and $\bl$.
In \Cref{sec:proofs-to-strategies} we prove one direction of the correspondence, translating $\eldt$ and $\elndt$ proofs to $\dbl$ or $\bl$ strategies, respectively.
In \cref{sec:strats-to-proofs-deterministic} we present the converse direction for the deterministic setting, translating $\dbl$ strategies to $\eldt$ proofs.
The remainder of the paper is devoted to establishing the same for the non-deterministic setting.

In \Cref{sec: Imm-szel} we present our formalisation of Immerman-\szel, witnessed by polynomial-size proofs in $\elndt$. 
Finally in \Cref{sec:strats-to-proofs} we use this to establish a translation from $\bl$ strategies to $\elndt$ proofs.
This final argument is quite involved, composed of several intermediate systems; we visualise its structure in \Cref{structure-of-paper}, and give some explanation in the caption.
The main difficulty is the bold cyan arrow (bottom), where we rely on the Immerman-\szel\ construction from \cref{sec: Imm-szel}.

\section{Proof systems for (non-deterministic) branching programs}
\label{sec:prelims-elndt}

In this section we shall briefly recall the systems $\eldt $ and $\elndt$, from \cite{DBLP:conf/csl/BussDasKnop,BussDasKnop19:preprint} and also appearing in \cite{Das-Del,DasDel25:pos-bps-journal},  reasoning about deterministic and non-deterministic branching programs respectively.

Throughout this work, we make use of \emphasis{propositional variables}, written $p,q,$ etc.
An \emphasis{assignment} is a map from propositional variables to Booleans $\{0,1\}$.
A \emphasis{Boolean function} is a map from assignments to $\{0,1\}$. 

 In proof complexity, formally, a (sound) \emphasis{propositional proof system} is just a polynomial-time function $P$ from $\Sigma^*$ to the set of propositional tautologies, for $\Sigma$ some finite alphabet.
 The idea is that
 $P$ checks (efficiently) that an element $ \sigma \in \Sigma^*$ correctly codes a proof in the system in which case the output $P(\sigma)$  is the tautology $\sigma$ proves.  Otherwise $P$ outputs the tautology 1 by convention.

A proof system $P$ \emphasis{polynomially simulates} a system $Q$ if we can construct in polynomial-time, from a $Q$-proof of $A$, a $P$-proof of $A$.
As is common in proof complexity, we will typically refrain from calculating explicit polynomial bounds throughout this work, which will always be evident from the arguments at hand.
Furthermore, while we often only state the \emph{existence} of small proofs, all of our arguments are feasibly constructive and such results comprise bona fide polynomial simulations.

A more comprehensive introduction to proof complexity can be found in the textbooks \cite{krajicek1995bounded,cooknguyen2010logical,krajivcek2019proof}.

\subsection{(Non-)deterministic branching programs}
A \emphasis{(non-deterministic) branching program (NBP)} is a (rooted) directed acyclic graph $G$ such that: 
\begin{itemize}
\item $G $ has two distinguished \emphasis{sink} nodes, $0$ and $1$.
\item $G$ has a unique \emphasis{root} node, i.e.\ a unique node with in-degree $0$.
    \item Each non-sink node of $G$ is labelled by a propositional variable.
    \item Each edge of $G$ is labelled by either $0$ or $1$.
\end{itemize}
We say that a NBP is \emphasis{deterministic} (a \emphasis{BP}) if each non-sink node has \emph{exactly} two outgoing edges, one labelled $0$ and the other labelled $1$.

A \emphasis{run} of a NBP $G$ on an assignment $\alpha$ is a maximal path beginning at the root of $G$ consistent with $\alpha$. 
I.e., at a node labelled by $p$ the run must follow an edge labelled by $\alpha(p) \in \bool$.
$G$ \emphasis{accepts} $\alpha$ if there is a run on $\alpha$ reaching the sink $1$.
We extend $\alpha$ to a map from all NBPs to $\bool$ by setting $\alpha(G) = 1 $ iff $G$ accepts $\alpha$.
Thus
each NBP \emphasis{computes} a unique Boolean function $\alpha \mapsto \alpha(G)$.

Further background on (non-deterministic) branching programs can be found in, e.g., the textbook \cite{Weg00:bps-and-bdds}.

\subsection{Representation via formulas with extension}

In order to syntactically represent (N)BPs, in addition to the propositional variables we fixed earlier we will also make use of a disjoint set of \emphasis{extension variables}, written $e_0,e_1,\dots$.
We briefly recall the syntax for expressing and reasoning about (non-deterministic) branching programs from \cite{DBLP:conf/csl/BussDasKnop,BussDasKnop19:preprint,Das-Del,DasDel25:pos-bps-journal}.
We refer the reader to those works for further examples and foundational details.

 \emph{Extended non-deterministic decision-tree} formulas, or simply \emphasis{\exndt-formulas}, written $A,B$ etc., are generated by the grammar:
\[
A,B \quad \bnf \quad 
0 \ \mid \ 1\ \mid \ ApB \ \mid\ A \lor B \ \mid \  \ e_i
\]

When writing formulas we employ some usual syntactic simplifications, omitting external and internal brackets when there is no ambiguity.
The \emphasis{size} of a formula $A$, written $|A|$, is the number of symbols occurring in $A$.

We often identify the propositional variable $p$ with the formula $0p1$.
\emphasis{\exdt-formulas} are \exndt-formulas not using $\lor$, and \nndt-formulas are \exnndt-formulas not using extension variables.  

\begin{remark}
    [Decision variables]
    Note that, since extension variables are disjoint from propositional variables, $Ae_iB$ is never a correct \exndt-formula; this will turn out to be important for our syntactic representation of (N)BPs to be faithful.
\end{remark}

Semantically $\lor$ is interpreted as usual disjunction, while $\dec A p B$ is interpreted as a decision ``if $p$ then $B$ else $A$''.  
Thus formulas without extension variables compute Boolean functions as usual.
Call such a formula \emphasis{valid} if it returns $1$ on all assignments.
The following result renders systems reasoning about even extension-free formulas suitable for proof complexity analysis:

\begin{propC}
[\cite{DBLP:conf/csl/BussDasKnop,BussDasKnop19:preprint}]
    \label{coNP-completeness}
    The set of valid \ndt\ formulas (even disjunctions of \dt\ formulas) is $\coNP$-complete.
\end{propC}

\noindent 
 The interpretation of extension variables is parametrised by a set of {`extension axioms'}. These axioms give meaning to extension variables by rendering them abbreviations of more complex formulas.
 Viewing (N)BPs as graphs, extension variables are essentially used to `name' the nodes.

\begin{definition}
[Extension axioms]
\label{extension-axioms-definition}
A set of \emphasis{extension axioms} $\E$ is a set of the form $ \{ \ee{i}{} \extiff E_i \}_{i<n} $, where each formula $E_i$ may only contain extension variables among $\ee{0}{}, \dots, \ee{i-1}{}$.  
\end{definition}

 \begin{figure}[t]
    \[
    \raisebox{-0.5\height}{
\begin{tikzpicture}[scale=.8, auto,swap]
\foreach \pos/\name/\disp in {
  {(0,4)/1/$p_1$}, 
  {(-1,3)/2/$p_2$},
  {(1,3)/3/$p_2$}, 
  {(-2,2)/4/$p_3$}, 
  {(0,2)/5/$p_3$},
  {(2,2)/6/$p_3$}, 
  {(-3,1)/8/$p_4$},
  {(-1,1)/10/$p_4$},
  {(1,1)/11/$p_4$},
  {(3,1)/12/$p_4$},
  {(-4,-0)/15/$0$},
  {(-2,-0)/16/$0$},
  {(0,-0)/17/$1$},
  {(+2,-0)/18/$0$},
  {(+4,-0.008)/19/$0$}}
\node[minimum size=12pt,inner sep=0pt] (\name) at \pos {\disp};

    \draw [->][thick,dotted](1) to (2);
    \draw [->][thin](1) to (3);
    
    \draw [->][thick,dotted](2) to (4);
    \draw [->][thin](2) to (5);
     
    \draw [->][thick,dotted](3) to (5);
    \draw [->][thin](3) to (6);

     \draw [->][thick,dotted](4) to (8);
     \draw [->][thin](4) to (10);
      \draw [->][thick,dotted](5) to (10);
    \draw [->][thin](5) to (11);
     \draw [->][thick,dotted](6) to (11);
     \draw [->][thin](6) to (12);
     \draw [->][thin](8) to (16);
     \draw [->][thick,dotted](8) to (15);
      \draw [->][thin](10) to (17);
     \draw [->][thick,dotted](10) to (16); \draw [->][thin](11) to (18);
     \draw [->][thick,dotted](11) to (17); \draw [->][thin](12) to (19);
     \draw [->][thick,dotted](12) to (18);
     \draw [->][thin]
    (1) [out=180, in=100] to  (2); \draw [->][thin]
    (2) [out=180, in=100] to  (4); \draw [->][thin]
    (4) [out=180, in=100] to  (8); \draw [->][thin]
    (8) [out=180, in=100] to  (15); \draw [->][thin]
    (3) [out=180, in=100] to  (5); \draw [->][thin]
    (5) [out=180, in=100] to  (10); \draw [->][thin]
    (10) [out=180, in=100] to  (16); \draw [->][thin]
    (6) [out=180, in=100] to  (11); \draw [->][thin]
    (11) [out=180, in=100] to  (17); \draw [->][thin]
    (12) [out=180, in=100] to  (18);
\end{tikzpicture}
}
\qquad
{
\begin{array}{r@{\ \extiff \ }l}
    \ee 1 1 & \posdec {\ee 2 1 } {p_1}{\ee 2 2} \\
    \ee 2 1 & \posdec {\ee 3 1 } {p_2}{\ee 3 2} \\
    \ee 2 2  & \posdec {\ee 3 2 } {p_2}{\ee 3 3} \\
    \ee 3 1  & \posdec {\ee 4 1 } {p_3}{\ee 4 2} \\
    \ee 3 2  & \posdec {\ee 4 2 } {p_3}{\ee 4 3} \\
    \ee 3 3  & \posdec {\ee 4 3 } {p_3}{\ee 4 4} \\
    \ee 4 1  & \posdec 0 {p_4} 0 \\
    \ee 4 2  & \posdec 0 {p_4} 1 \\
    \ee 4 3  & \posdec 1 {p_4} 0 \\
    \ee 4 4  & \posdec 0 {p_4} 0
\end{array}
}
\]
\vspace{-3ex}
\caption{NBP for 2-out-of-4 threshold and representation by extension axioms.
Here $0$-edges are dotted, and $1$-edges are solid; the multiple $0$-leaves correspond to the same sink.
}
\label{fig:thresh-nbp-4-2}
\end{figure}

\begin{example}
[2-out-of-4 threshold]
\label{example-2-4-thresh}
The 2-out-of-4 threshold function, returning 1 iff at least two of its four inputs are 1, is computed by the NBP on the left of Fig.~\ref{fig:thresh-nbp-4-2}. 
This NBP may be represented by the extension variable $\ee 1 1 $ under the extension axioms on the right of Fig.~\ref{fig:thresh-nbp-4-2}.
Each $\ee i j$ represents the $j$\textsuperscript{th} node (left to right) on the $i$\textsuperscript{th} row (top to bottom) for $1\leq i\leq 4$ and $1 \leq j\leq i$. 

For well-foundedness of the extension axioms (i.e.\ the subscripting condition of \cref{extension-axioms-definition}), note that we may identify  each $\ee ij$ with $\ee {4(4-i) + j}{}$. 
We shall typically leave such identifications implicit in the sequel.
\end{example}

Note in \cref{extension-axioms-definition} the condition that each $e_i$ `abbreviates' only formulas with extension variables of smaller index.
This means that the graphs described by formulas under a set of extension axioms is well-founded, i.e.\ a dag, inducing a natural induction principle:

\begin{remark}
[$\mathcal E$-induction]
\label{A-induction}
Given a set of extension axioms $\mathcal E = \{ \ee{i}{} \extiff E_i \}_{i<n} $ we may define a strict partial order $<_\mathcal E$ on formulas over $\ee{0}{},\dots, e_{n-1}$ by:
\begin{itemize}
    \item $A <_\mathcal E \dec A p B$ and $B <_\mathcal E \dec A p B$.
    \item $A <_\mathcal E A \lor B$ and $B <_\mathcal E A \lor B$.
    \item $E_i <_\mathcal E \ee{i}{}$, for each $i<n$.
\end{itemize}
Observe that $<_\mathcal E$ is indeed well-founded by the condition that each $E_i$ must contain only extension variables among $\ee{0}{}, \dots, \ee{i-1}{}$.
Thus we may carry out arguments and make definitions by induction on $<_\mathcal E$, which we shall simply refer to as \emphasis{$\mathcal E$-induction}. 
\end{remark}

From here, for instance, we obtain:

\begin{definition}
[Semantics of \exndt\ formulas]
\label{dfn:sat-\exndt-wrt-extax}
The semantics of \exndt\ formulas with respect to a set of extension axioms $\mathcal E = \{\ee{i}{}{}{} \extiff E_i \}_{i<n}$, are given by $\sat{\mathcal E}$,  an infix binary relation between assignments and formulas over $\ee{0}{}, \dots, \ee{n-1}{}{}{}$ defined by $\mathcal E$-induction by:
\begin{itemize}
\item $\alpha \notsat {\mathcal E} 0$ and $\alpha \sat {\mathcal E} 1$.
        \item $\alpha \sat {\mathcal E} \dec A p B$ if either $\alpha(p)=0$ and $\alpha \sat {\mathcal E} A$, or $\alpha(p)=1$ and $\alpha \sat {\mathcal E} B$. 
    \item $\alpha \sat {\mathcal E} A \lor B$ if $\alpha \sat {\mathcal E} A$ or $\alpha \sat {\mathcal E} B$.
    \item $\alpha \sat {\mathcal E} \ee{i}{}{}{}$ if $\alpha \sat {\mathcal E} E_i$.
\end{itemize}
If $ \alpha \sat {\mathcal E} A \iff f(\alpha) = 1$  for some Boolean function $f$, we say that $A$ \emphasis{computes $f$ over} $\mathcal E$.
We may also say that $A$ is \emphasis{valid over $\mathcal E$} if, for every assignment $\alpha$ we have $\alpha \sat{\mathcal E} A$ (equivalently, $A$ computes the constant function $1$ over $\E$).
\end{definition}

Since many of our arguments are based on $\mathcal E$-induction, let us recall the following complexity analysis from earlier work: 
\begin{proposition}
[Complexity of $\mathcal E$-induction, e.g.\ \cite{DasDel25:pos-bps-journal}]
\label{complexity-of-A-induction}
Let $A$ be an \exdt\ or \exndt\ formula over a set of extension axioms $\mathcal E = \{ \ee{i}{}{}{} \extiff E_i\}_{i<n} $.
Then $| \{ B<_\mathcal E A\}| \leq |A| + \sum\limits_{i<n} |E_i|$ and, if $B<_\mathcal E A$, then $|B|\leq \max (|A|,|E_0|, \dots, |E_{n-1}|)$.
\end{proposition}

\begin{figure}
    \textbf{Identity and cut:}
    \smallskip
    \[   
    \vlinf{\id}{}{A \seqar              A}{}
    \qquad
    \vliinf{\cut}{}{\Gamma \seqar                 \Delta}{\Gamma \seqar                 \Delta, A}{\Gamma, A \seqar                 \Delta}
    \]
    
    \medskip
    
    \textbf{Structural rules:}
    \smallskip
    \[
    \begin{array}{c@{\qquad}c@{\qquad}c}
        \vlinf{\lefrul\exch}{}{\Gamma, B,A ,\Gamma'\seqar \Delta}{\Gamma, A, B ,\Gamma' \seqar \Delta}  
        &
        \vlinf{\lefrul \wk}{}{\Gamma, A \seqar                 \Delta}{\Gamma \seqar                 \Delta}
        &
        \vlinf{\lefrul\cntr}{}{\Gamma, A \seqar                 \Delta}{\Gamma, A, A \seqar                 \Delta}
        \\
        \noalign{\smallskip}
        \vlinf{\rigrul\exch}{}{\Gamma\seqar \Delta, B,A ,\Delta'}{\Gamma \seqar \Delta, A, B , \Delta'}
         & 
         \vlinf{\rigrul \wk}{}{\Gamma \seqar                 \Delta, A}{\Gamma \seqar                 \Delta}
         &
         \vlinf{\rigrul\cntr}{}{\Gamma \seqar                 \Delta, A}{\Gamma \seqar                 \Delta, A, A}
    \end{array}
    \]

\medskip

\textbf{Logical rules:}
\smallskip
\[
\begin{array}{c@{\quad}c@{\quad}c@{\quad}}
      \vlinf{ \lefrul 0}{}{0 \seqar                 }{} &
        \vlinf{\lefrul 1 }{}{\Gamma, 1 \seqar \Delta}{\Gamma \seqar \Delta} 
     & 
     \vliinf{\lefrul \lor}{}{\Gamma, A \lor B \seqar                 \Delta}{\Gamma , A \seqar                 \Delta}{\Gamma, B \seqar                 \Delta}
     \\
     \vlinf{\rigrul 0}{}{\Gamma \seqar \Delta,0}{\Gamma\seqar \Delta}
     &
     \vlinf{ \rigrul 1}{}{\seqar                 1}{}
     &
     \vlinf{\rigrul \lor}{}{\Gamma \seqar                 \Delta, A\lor B}{\Gamma \seqar                 \Delta, A,B}
\end{array}
\]
\[
   \vliinf{\lefrul p}{}{\Gamma, \dec A p B \seqar                 \Delta}{\Gamma, A \seqar                 \Delta, p}{\Gamma, p, B \seqar                 \Delta}
   \qquad
 \vliinf{\rigrul p}{}{\Gamma \seqar                 \Delta, \dec A p B}{\Gamma \seqar                 \Delta, A, p}{\Gamma, p \seqar                 \Delta , B} 
\]

    \caption{Rules for systems $\mathsf{(e)}\ldt$ and $\mathsf{(e)}\lndt$.}
    \label{fig:lndt}
\end{figure}

\subsection{Proof systems}

\label{section: proof systems for NBPs}

We now recall the sequent systems $\eldt$ and $\elndt$ reasoning over \exdt\ and \exndt\ formulas, respectively, first introduced in \cite{BussDasKnop19:preprint,DBLP:conf/csl/BussDasKnop}.\footnote{The qualifier $\mathsf L$ here is just a standard indication for a sequent system over the corresponding logic, e.g.\ $\mathsf{LK}$ for classical logic and $\mathsf{LJ}$ for intuitionistic logic, originating from Gentzen's works.} 
Lines in these systems represent deterministic and non-deterministic branching programs respectively.

A \emphasis{sequent} is an expression $\Gamma \seqar \Delta$, where $\Gamma$ and $\Delta $ (called \emphasis{cedents})
are lists of formulas (`$\seqar$' is just a syntactic delimiter).
Such a sequent is interpreted as ``if all formulas in the antecedent $\Gamma$ are true then some formula of the succedent  $\Delta$ is true''.

\begin{definition}
[Systems for \exnndt\ formulas]
\label{eldt-and-elndt-systems}
The system $\lndt$ is given by the rules in \cref{fig:lndt}.
  $\ldt$ is the restriction of $\lndt$ not including rules for disjunction, $\lor$. 
  
  An $\lndt$ ($\ldt$) \emphasis{derivation} of $\Gamma \seqar                 \Delta$ from \emphasis{hypotheses} $\mathcal H = \{\Gamma_i \seqar                 \Delta_i\}_{i\in I}$ is defined as expected: it is a finite list of sequents, each either some $\Gamma_i \seqar                 \Delta_i$ from $\mathcal H$ or following from previous ones by rules of $\lndt$ ($\ldt$, resp.) and ending with $\Gamma \seqar   \Delta$.
An $\elndt$ ($\eldt$) \emphasis{proof} of a sequent $\Gamma \seqar \Delta$ over a set of ($\lor$-free, resp.) extension axioms $\E$ is an $\lndt$ ($\ldt$, resp.) derivation $\pi$ with hypotheses $\E$.\footnote{$A \extiff B$ is construed as an abbreviation for the pair of sequents $A \seqar B $ and $B \seqar A$.}
If $\Gamma\seqar \Delta$ does not have extension variables we just say that $\pi$ is an $\elndt$ (or $\eldt$, resp.) proof of $\Gamma \seqar \Delta$.

 The \emphasis{size} of a  derivation $\pi$, written $|\pi|$, is the number of symbols occurring in it.
\end{definition}


\begin{remark}
[Cedents as lists]
    As usual, it is not important whether cedents are construed as lists, multisets or sets for studying proof complexity of $\eldt$ and $\elndt$.
    However working with lists accommodates more faithfully the definition of long disjunctions and conjunctions later when we work with games.
\end{remark}


\begin{proposition}
    [Soundness and completeness, \cite{DBLP:conf/csl/BussDasKnop,BussDasKnop19:preprint}]
    \label{soundness-completeness}
    Suppose $\Gamma, \Delta $ are sequents of \nndt\ formulas.
    Then, respectively, $\lnndt$ proves $\Gamma \seqar \Delta$ iff 
    for every assignment $\alpha$, either $\alpha\not \models A$ for some $A \in \Gamma$ or $\alpha \models B$ for some $B\in \Delta$. 
\end{proposition}

\noindent
Note that, as the validity problem for disjunctions \dt\ formulas is $\coNP$-complete, cf.~\cref{coNP-completeness}, so is the provability of \dt-sequents. 
In fact we can establish a stronger version of the above result, relativised to a set of extension axioms, though we do not ever rely on this.
Nonetheless some of our proof translations will be parametrised by sets of extension axioms.

\begin{example}
\label{ex:eij-in-eik}
    Recalling \cref{example-2-4-thresh}, let us write $\E$ for the set of extension axioms in \cref{fig:thresh-nbp-4-2}, right.
    Over this set we have $\elndt $ proofs of, say, $\ee 4 2 \seqar \ee 4 3$ by,
    \begin{equation}
        \label{eq:e42-imp-e43}
    \vlderivation{
    \vlin{\lefrul \E}{}{\ee 4 2 \seqar \ee 4 3}{
    \vliin{\lefrul{p_4}}{}{\posdec 0 {p_4} 1 \seqar \ee 4 3 }{
        \vliq{\wk}{}{0 \seqar \ee 4 3 , p_4}{
        \vlin{\lefrul 0 }{}{0 \seqar }{\vlhy{}}
        }
    }{
        \vlin{\rigrul \E}{}{p_4, 0\lor 1 \seqar \ee 4 3}{
        \vliin{\rigrul {p_4} }{}{p_4, 0\lor 1 \seqar \posdec 1 {p_4} 0}{
            \vliq{\wk}{}{p_4 , 0 \lor 1 \seqar 1, p_4}{
            \vlin{\rigrul 1}{}{\seqar 1}{\vlhy{}}
            }
        }{
            \vlin{\rigrul \lor }{}{p_4,0\lor 1 \seqar 1 \lor 0}{
            \vliq{\wk}{}{p_4, 0 \lor 1 \seqar 1,0}{
            \vlin{\rigrul 1}{}{\seqar 1}{\vlhy{}}
            }
            }
        }
    }
    }
    }
    }
    \end{equation}
    where, formally, the steps $\lefrul \E$ and $\rigrul \E$ are obtained by `cutting' against the appropriate extension hypotheses.
    We have also omitted several `contraction' and `exchange' steps, essentially treating each side of the sequent as \emph{sets} of formulas instead of lists for convenience.
    Note that, while we have displayed the proof above as a tree, proofs may in general be \emph{dags}, so that when measuring proof complexity it suffices to count only the number of symbols in \emph{distinct} sequents. E.g.\ we may \emph{identify} the two initial sequents $\seqar 1$ above.

       \newcommand{\storage}{\eqref{eq:e42-imp-e43}}
    
   We can use this to build further proofs, e.g.\ of $\ee 3 2 \seqar \ee 33$: 
   \begin{equation}
       \label{eq:e32-in-e33}
    \small
    \fontdimen6\textfont2=1em 
    \vlderivation{
    \vlin{\E}{}{\ee 3 2 \!\seqar\! \ee 33 }{
    \vliin{\lefrul{p_3}}{}{\posdec {\ee 4 2 }{p_3}{\ee 4 3} \!\seqar\! \posdec {\ee43}{p_3}{\ee44}}{
        \vliin{\rigrul{p_3}}{}{\ee 4 2 \!\seqar\! p_3, \posdec{\ee 4 3 }{p_3}{\ee 4 4 }}{
            \vlin{\rigrul \wk}{}{\ee 4 2 \!\seqar\! p_3, \ee 43 }{
            \vliq{}{}{\ee 4 2 \!\seqar\! \ee 43}{\vlhy{\text{\storage}}}
            }
        }{
            \vliq{\wk}{}{\ee 4 2 , p_3 \!\seqar\! p_3 \ee 4 3 \lor \ee 44 }{
            \vlin{\id}{}{p_3\!\seqar\! p_3}{\vlhy{}}
            }
        }
    }{
        \vliin{\rigrul{p_3}}{}{p_3 , \ee 4 2 \lor \ee 4 2 \!\seqar\! \posdec {\ee 43}{p_3}{\ee44}}{
            \vliq{\wk}{}{p_3, \ee42 \lor \ee43 \!\seqar\! \ee43,p_3}{
            \vlin{\id}{}{p_3 \!\seqar\! p_3}{\vlhy{}}
            }
        }{
            \vlin{\rigrul \lor}{}{p_3 , \ee42 \lor \ee43 \!\seqar\! \ee43\lor\ee44}{
            \vliq{\wk}{}{p_3, \ee42\lor\ee43 \!\seqar\! \ee43,\ee44}{
            \vliin{\lefrul\lor}{}{\ee42\lor \ee43 \!\seqar\! \ee43}{
                \vliq{}{}{\ee42\!\seqar\!\ee43}{\vlhy{\text{\storage}}}
            }{
                \vlin{\id}{}{\ee43\!\seqar\!\ee43}{\vlhy{}}
            }
            }
            }
        }
    }
    }
    }
   \end{equation}
   More generally we can build $\elndt$ proofs of all $\ee i j \seqar \ee i k$ for $j\leq k$.
\end{example}



\begin{example}
We can build a $\lndt$ proof of $\dec {(\dec A p B)} q D \seqar \dec A q D \lor p$ as follows:
\[
\small
\vlderivation{
\vlin{\lor}{}{\dec {(\dec A p B)} q D \seqar \dec A q D \lor p}{
        \vliin{\lefrul q}{}{\dec {(\dec A p B)} q D \seqar \dec A q D , p}{
            \vliin{\lefrul p}{}{\dec A p B \seqar \dec A q D, q, p}{
                \vliin{\rigrul q}{}{A \seqar \dec A q D , q , p}{
                    \vliq{\wk}{}{A \seqar A , q, p}{
                    \vlin{\id}{}{A \seqar A}{\vlhy{}}
                    }
                }{
                    \vliq{\wk}{}{q \seqar D,q,p}{
                    \vlin{\id}{}{q \seqar q}{\vlhy{}}
                    }
                }
            }{
                \vliq{\wk}{}{p,B \seqar \dec A q D , q, p}{
                \vlin{\id}{}{p\seqar p}{\vlhy{}}
                }
            }
        }{
            \vliin{\rigrul q}{}{q,D\seqar \dec A q D,p}{
                \vliq{\wk}{}{q,D\seqar A,q,p}{
                \vlin{\id}{}{q \seqar q}{\vlhy{}}
                }
            }{
                \vliq{\wk}{}{q,D \seqar D,p}{
                \vlin{\id}{}{D \seqar D}{\vlhy{}}
                }
            }
        }
}
}
\]
This has an $\ldt$ subproof of $\dec {(\dec A p B)} q D \seqar \dec A q D , p$.
We can similarly build an $\ldt$ proof of $\dec {(\dec A p B)} q D , p \seqar \dec B q D$. 
By combining these two as below we can obtain the `distribution' property $\dec {(\dec A p B)} q D \seqar \dec {(\dec A q D)} p {(\dec B q D)}$ in $\ldt$:
\[
\vliinf{\rigrul p}{}{\dec {(\dec A p B)} q D \seqar \dec {(\dec A q D)} p {(\dec B q D)}}{\dec {(\dec A p B)} q D \seqar \dec A q D , p}{\dec {(\dec A p B)} q D , p \seqar \dec B q D}
\]
\end{example}

\section{Prover-Adversary games for (non-deterministic) branching programs}
\label{sec:games}

 The goal of this section is to present games that correspond to $\elnndt$ in the same way that the Boolean formula game of Pudl\'ak and Buss corresponds to Frege \cite{PudBuss}.
 The basic idea of these games is simple.
 Two players, \emphasis{Prover} and \emphasis{Adversary} alternate turns as follows: Prover asks `queries', and {Adversary} assigns them Boolean values. 
 Prover wins if the set of Adversary's answers ever include a `simple contradiction'. 
Such games are duly parametrised by their notions of `query' and `simple contradiction'.
 Before presenting our games,
 we first need to develop some further machinery for defining these data.

\subsection{`Similar' representations of branching programs}
Since (N)BPs have underlying dag structure that we represent via extension, we must address the matter of \emph{equivalent} representations.
While we can mostly sidestep this issue for $\eldt $ and $\elndt$, with different representations being demonstrably equivalent by way of polynomial-size proofs, in the setting of games we must be more resource-conscious of the number of rounds it takes to win (which should be logarithmic).
For this reason 
we shall `bootstrap' our games to include native contradictory specifications, namely 
for a notion of \emph{simulation} of transition systems specialised to our syntax.
In this way, we can really see our games as operating \emph{directly} on (N)BPs rather than their representations, equating bisimilar programs.

We define a judgement $A\simulates \E B$  meaning that the NBP represented by $A$ (over $\E$) \emph{simulates} the NBP represented by $B$ (over $\E$), as transition systems.


\begin{definition}
[Simulation]
    Let $\E = \{e_i\dseqar E_i \}_{i<n}$ be a set of extension axioms.
    We define the judgement $A\simulates\E B$, read `\emphasis{$A$ $\E$-simulates $B$}', by:
    \begin{equation}
        \label{eq:simulation-rules}
    \begin{array}{c@{\qquad}c@{\qquad}c}
        \vlinf{}{}{A \simulates \E A}{} 
        &
        \vliinf{}{}{A\simulates \E C\lor D}{A\simulates \E C}{A\simulates \E D}
        &
        \vlinf{}{}{A\simulates \E e_i}{A\simulates \E E_i}
        \\
        \noalign{\smallskip}
          \vliinf{}{}{ApB \simulates \E CpD}{A\simulates\E C}{B\simulates \E D}
         & 
    \vlinf{}{}{A_0 \lor A_1 \simulates \E B}{A_i \simulates \E B}
    &
    \vlinf{}{}{e_i \simulates \E B}{E_i \simulates \E B}
    \end{array}
     \end{equation}
\end{definition}

\begin{remark}    
It is not difficult to see that $A\simulates \E B$ if and only if the NBP represented by $A$ (over $\E$) \emph{simulates} the NBP represented by $B$ (over $\E$), as labelled transition systems (see e.g.\ \cite[Section~7.4]{book:principles-model-checking} for definitions).  
The forwards implication is by induction on the definition of $\simulates\E$, while the backwards implication follows
by $\E$-induction on $B$ then $A$.
\end{remark}


\begin{proposition}
[Simulation in $\elndt$]
\label{simulation-has-polysize-proofs}
\label{simulation result}
Let $\E$ be a set of \exndt\ extension axioms and
suppose $A\simulates\E B$.
There are polynomial-size $\elndt$ proofs of $B\seqar A$ over $\E$.
If $\E,A,B$ are $\lor$-free then these proofs are furthermore in $\eldt$.
\end{proposition}

\begin{proof}
    We proceed by induction on the definition of simulation, deriving each rule of \cref{eq:simulation-rules} in $\elndt$, noticing that we do not need to use $\lor$-steps when $\E,A,B$ are $\lor$-free.
    Almost every step is routine, requiring only a couple steps in $\elndt$, except for the lower left rule of \cref{eq:simulation-rules}, which is derived by,
     \[
    \vlderivation{
    \vliin{\rigrul p}{}{CpD \seqar ApB}{
        \vliin{\lefrul p}{}{CpD \seqar A,p}{
            \vlin{\wk}{}{C\seqar A,p,p}{
            \vliq{\IH}{}{C\seqar A}{\vlhy{}}
            }
        }{
            \vlin{\wk}{}{p,D \seqar A,p}{
            \vlin{\id}{}{p\seqar p}{\vlhy{}}
            }
        }
    }{
        \vliin{\lefrul p }{}{CpD , p \seqar B}{
            \vlin{\wk}{}{C,p \seqar B,p}{
            \vlin{p}{}{p\seqar p}{\vlhy{}}
            }
        }{
            \vlin{\wk}{}{p,D , p \seqar B}{
            \vliq{\IH}{}{D\seqar B}{\vlhy{}}
            }
        }
    }
    }
    \]
    where the steps marked $\IH$ are obtained by the inductive hypothesis.
\end{proof}

Note that, in the case of deterministic BPs, simulation reduces to the simpler notion of equivalence between branching programs with the same `unfolding' as decision trees.

\begin{definition}
    [Unfolding]
    Let $\E = \{e_i \extiff E_i\}_{i<n}$ be a set of $\lor$-free extension axioms and $A$  a eDT formula.
    We define the DT-formula $\unf A \E$ by $\E$-induction on $A$:
    \begin{itemize}
        \item $\unf 0 \E \dfn 0$
        \item $\unf 1 \E \dfn 1$
        \item $\unf {ApB} \E \dfn \unf A \E p \unf B \E $
        \item $\unf {e_i}\E \dfn \unf{E_i}\E $, for $i<n$.
    \end{itemize}
\end{definition}
\noindent 
In the presence of disjunction the natural extension of the notion of unfolding above is not canonical: the same NBP may be represented by different bracketings of $\lor$ under associativity and commutativity.
This is why we must work with the notion of simulation above in the games we introduce later.
In any case, for deterministic BPs:

\begin{proposition}\label{simulation and unfolding}
    Let $\E$ be a set of $\lor$-free extension axioms and $A,B$ eDT formulas.
    $A\simulates \E B$ if and only if $\unf A \E = \unf B \E$.
\end{proposition}
\begin{proof}
[Proof sketch]
The $\Leftarrow$ direction is routine, following by $\E$-induction on $A$ and $B$. For the $\Rightarrow$ direction we proceed by induction on the definition of simulation.
Since $\E, A,B$ are $\lor$-free, no disjunction cases of the definition of $\simulates \E$ apply. 
All the remaining cases follow immediately from inductive hypothesis and the definition of $\unf  - \E$.
\end{proof}

\medskip

\subsection{Boolean combinations of branching programs.}

The next matter we must address is to some extent a design choice.
To prove equivalence of our games and $\elnndt$, recall that games are essentially tree-like versions of the inference system they correspond to, so we shall need to prove some sort of closure under Boolean combinations, cf.~\cite{krajivcek1994lower}.
For $\elndt$, this will amount to the formalisation of a (non-uniform) version of Immerman-\szel. 

We could carry out this work either within a game system or within an inference system.
However doing so within games requires us to again be resource conscious of the number of rounds, which must be logarithmic, and it is not clear to us that this can be duly carried out without further bootstrapping.
Thus we simply expand the queries of our game to be as expressive as possible, closing them under Boolean combinations.
This has the effect of simplifying the translation from $\elnndt$ proofs to strategies, but rendering the converse more difficult.

\begin{definition}
    [Boolean combinations]
    Write $P,Q,R$ etc.\ for \emphasis{Boolean combinations} of \exnndt\ formulas (\emphasis{\boolexnndt-formulas}), generated by,
    \[
    P,Q,R,\dots 
    \quad ::= \quad 
    A \ \mid \ \cnot Q \ \mid \ Q \lor R \ \mid \ Q \cand R
    \]
    where $A$ ranges over \exnndt\ formulas.
\end{definition}
%
%

The intended semantics of this extended class of formulas are as expected, again parametrised by a set of extension axioms in order to interpret the extension variables in a \exnndt-subformula.
Note that the syntax here is `two-tiered': Boolean connectives may not alternate with decisions and extensions (except $\lor$ in the case of \boolexndt). 
This is reflected by the use of different metavariables ($P,Q,R$ etc.) for Boolean combinations of branching programs.
Later we consider intermediate systems that extend $\elnndt$ by (positive) Boolean combinations, interpolating the translation from strategies to $\elnndt$.

\medskip

\subsection{Games for (non-deterministic) branching programs.}
Finally we are in a position to define our games for deterministic and non-deterministic branching programs, ultimately corresponding to $\eldt$ and $\elndt$ respectively.

A \emphasis{(Prover-Adversary) game} is given by a set of \emphasis{queries} and a set of \emphasis{(simple) contradictions}, which are sets of Boolean assignments to queries.
From here the mechanics of the Prover-Adversary game are defined as in \cite{PudBuss}: beginning with an initial set of assignments $S$ (a \emphasis{state} or \emphasis{specification}) Prover asks queries, and Adversary assigns them a Boolean value.
If during a play the set of assignments accumulated (including the initial ones) includes a simple contradiction then Prover wins.
A \emphasis{winning strategy} $\sigma$ (for Prover) from $S$ is represented as a full binary tree in the natural way, with queries labelling non-leaf nodes, as in \Cref{strategy} for the game we are about to define.
Writing $\Gamma \mapsto b$ for $\{A \mapsto b : A \in \Gamma\}$, we call a winning strategy from state $\{\Gamma \mapsto 1, \Delta \mapsto 0\}$ a \emphasis{proof} of $\Gamma \seqar \Delta$, with respect to the underlying game.

\begin{figure}[t]
    \[
{\begin{tikzpicture}[scale=1,grow=up,every tree node/.style={draw=none},
   level distance=1cm,sibling distance=.10cm,
   edge from parent path={(\tikzparentnode) -- (\tikzchildnode)}]
\Tree
[.{\small $\{\neg \neg(Q\lor \neg Q ) \mapsto 0\}$} 
\edge[dashed] node[right] {};
[.{\small $\neg(Q\lor \neg Q )$}
    \edge node[pos=.68,auto=right] {$1$};
    [.{\small$Q\lor \neg Q $} 
    \edge node[pos=.68,auto=right] {$1$};
    [.$\text{s.c.}\neg $ ]
    \edge node[pos=.68,auto=left] {$0$};
    [.{\small $\neg Q$}
    \edge node[pos=.68,auto=right] {$1$};
    [.{\small$ \text{s.c.} \lor $} 
   ]
    \edge node[auto=left] {$0$};
    [. {\small $Q$ } 
     \edge node[pos=.68,auto=right] {$1$};
    [.{\small$ \text{s.c.} \lor $} 
   ]
    \edge node[auto=left] {$0$};
    [. {\small $\text{s.c.}\neg$ } 
     ]   ]  ]]
    \edge node[auto=left] {$0$};
    [.  $\text{s.c.}\neg$  ]]]
\end{tikzpicture}}
\qquad
{\begin{tikzpicture}[scale=1,grow=up,every tree node/.style={draw=none},
   level distance=1cm,sibling distance=.10cm,
   edge from parent path={(\tikzparentnode) -- (\tikzchildnode)}]
\Tree
[.{\small $\{\ee 4 2 \mapsto 1 , \ee 43 \mapsto 0\}$} 
\edge[dashed] node[right] {};
[.{\small $\posdec 0 {p_4} 1$}
    \edge node[pos=.68,auto=right] {$1$};
    [.{\small$p_4 $} 
    \edge node[pos=.68,auto=right] {$1$};
    [.{\small$\posdec{ 1} {p_4} {0} $ }
    \edge node[pos=.68,auto=right] {$1$};
     [.{\small$ \text{s.c.} \E $} ]
    \edge node[pos=.68,auto=left] {$0$};
    [.{\small $1 \lor 0$}
    \edge node[pos=.68,auto=right] {$1$};
    [.{\small$ \text{s.c.} p_4 $} 
   ]
    \edge node[auto=left] {$0$};
    [. {\small $0$ } 
     \edge node[pos=.68,auto=right] {$1$};
    [.{\small$ \text{s.c.} 0 $} 
   ]
    \edge node[auto=left] {$0$};
    [. {\small $0$ } 
     \edge node[pos=.68,auto=right] {$1$};
    [.{\small$ \text{s.c.} p_4 $} 
   ]
    \edge node[auto=left] {$0$};
    [. {\small $\text{s.c.}1$ } 
     ]   ]  ]  ]
    ]
    \edge node[pos=.68,auto=left] {$0$};
    [.{\small $ 0$}
    \edge node[pos=.68,auto=right] {$1$};
    [.{\small$ \text{s.c.} 0 $} 
   ]
    \edge node[auto=left] {$0$};
    [.{\small$ \text{s.c.} p_4 $} 
   ]  ]]
    \edge node[auto=left] {$0$};
    [.{\small$\text{s.c.}\E$}  ]]]
\end{tikzpicture}}
\]
    \caption{Two winning strategies for Prover, over $\E$ the set of extension axioms from \cref{fig:thresh-nbp-4-2}. Initial states are indicated below the roots by dashed edges. Simple contradictions are indicated as leaves marked `s.c.' along with their type.}
    \label{strategy}
\end{figure}

We are finally ready to present our games corresponding to $\elnndt$:


\begin{definition}
    The \emphasis{non-deterministic branching program game} $\bl$ (over $\E= \{e_i \extiff E_i\}_{i<n}$) is given by:
    \begin{itemize}
        \item Queries are just the \boolexndt-formulas.
        \item Simple contradictions consist of just:
        \begin{itemize}
        \item Boolean contradictions: $\{0\mapsto 1\}$ and $\{1 \mapsto 0\}$.
        \item Decision contradictions:
        $\{A_0 \mapsto b_0, A_1 \mapsto b_1, p \mapsto i, A_0pA_1 \mapsto c : c \neq b_i\}  $.
            \item Boolean connective contradictions: all sets inconsistent with the truth tables for $\lnot,\lor,\land$:
            \begin{itemize}
                \item $\{Q\mapsto b, \lnot Q \mapsto b\}$
                \item $\{Q_0 \lor Q_1 \mapsto 0, Q_i \mapsto 1\}$ and $\{Q_0 \lor Q_1 \mapsto 1, Q_0 \mapsto 0, Q_1 \mapsto 1\}$
                \item $\{Q_0\land Q_1 \mapsto 0, Q_0 \mapsto 1, Q_1 \mapsto 1\}$ and $\{Q_0 \land Q_1 \mapsto 1, Q_i \mapsto 0\}$.
            \end{itemize}
            \item Extension contradictions: $\{e_i \mapsto b, E_i \mapsto 1-b\} $.
            \item Similarity contradictions:\footnote{Note that the similarity contradictions actually include the extension ones as a special case, but this serves as some useful redundancy.} $\{A \mapsto 0, B\mapsto 1\}  $ whenever $A\simulates \E B$.
        \end{itemize}
    \end{itemize}
    The \emphasis{deterministic branching game} $\dbl$ (wrt $\E$) is defined the same way, only wrt \boolexdt\ instead of \boolexndt.
\end{definition}

\begin{example}
    In \cref{strategy} we give two winning strategies for Prover, one from the state $\{\lnot \lnot (Q \lor \lnot Q) \mapsto 0\}$, left, and one from $\{\ee 4 2 \mapsto 1, \ee43 \mapsto 0\}$, right, over the set $\E$ of extension axioms from \cref{fig:thresh-nbp-4-2}.
    They are $\bl$ proofs of the sequents $\seqar \lnot \lnot (Q \lor \lnot Q) $ and $\ee 4 2 \seqar \ee43$ (over $\E$).
    The left strategy is also a $\dbl$ proof, as long as $Q$ is $\lor$-free.
\end{example}

\noindent
We shall continue to notate (winning) strategies similarly to \cref{strategy}.

\section{Proofs to strategies}
\label{sec:proofs-to-strategies}

In this section we present one direction of our main results, translating propositional proofs in $\eldt$ or $\elndt$ to strategies in $\dbl$ or $\bl$, respectively:

\begin{theorem}
\label{games-simulate-proofs}
If $\elndt$ (or $\eldt$) has a size $N$ proof of a sequent $\Gamma\seqar \Delta$ over $\E$, there is a $O(\log N)$ round winning strategy from $\{\Gamma \mapsto 1, \Delta \mapsto 0\}$ in $\bl$ (or $\dbl$, resp.) over $\E$.
\end{theorem}

\noindent
Also in this section we shall present a converse to this result in the deterministic setting.
All definitions and results in this section apply to both $\dbl$ and $\bl$, unless otherwise stated.

\subsection{A strategic toolbox}
Before presenting a proof of \cref{games-simulate-proofs} above, let us build up some machinery for building strategies.
When describing strategies, we shall use some suggestive terminology:
\begin{itemize}
  \item ``\emphasis{force} $Q\mapsto b$ (or win) from $S$ in $r$ rounds'' is a (partial) strategy from $S$ of depth $\leq r$ with each leaf either a simple contradiction or  has incoming edge $Q \mapsto b$;
    \item ``\emphasis{find} $\alpha\in T$ (or win) from $S$ in $r$ rounds'' is a (partial) strategy from $S$ of depth $\leq r$ with each leaf either a simple contradiction or has incoming edge labelled by some assignment in $T$, say $\alpha$.
\end{itemize}
We shall almost always omit `or win' when using these phrases.
We shall often expand out `$\alpha \in T$' according to the context, e.g.\ saying `find $Q \in \Gamma$ with $Q\mapsto b$'.

\begin{figure}[t]  
    \centering
\[
{\begin{tikzpicture}[scale=.9,grow=up,every tree node/.style={draw=none},
   level distance=1.38cm,sibling distance=.88cm,
   edge from parent path={(\tikzparentnode) -- (\tikzchildnode)}]
\Tree
[.{\small $\{Q_0\lor Q_1 \mapsto 1\}$} 
\edge[dashed] node[right] {};
[.{\small $Q_0$}
    \edge node[pos=.68,auto=right] {$1$};
    [.{\small .} 
   ]
    \edge node[auto=left] {$0$};
    [. {\small $Q_1$}
      \edge node[pos=.68,auto=right] {$1$};
    [.{\small .} 
   ]
    \edge node[auto=left] {$0$};
    [. {\small  s.c. $\lor$}
    ]]]] \end{tikzpicture}
    }\qquad \qquad {\begin{tikzpicture}[scale=.9,grow=up,every tree node/.style={draw=none},
   level distance=1.38cm,sibling distance=.88cm,
   edge from parent path={(\tikzparentnode) -- (\tikzchildnode)}]
    \Tree
[.{\small $\{\neg Q \lor R\mapsto 0\}$} 
\edge[dashed] node[right] {};
[.{\small $Q$}
    \edge node[pos=.68,auto=right] {$1$};
    [.{\small $R $} 
     \edge node[pos=.68,auto=right] {$1$};
    [.{\small s.c. $\lor$} 
   ]
    \edge node[auto=left] {$0$};
    [. {\small .}
    ]
   ]
    \edge node[auto=left] {$0$};
    [. {\small $R$}
      \edge node[pos=.68,auto=right] {$1$};
    [.{\small s.c. $\lor$} 
   ]
    \edge node[auto=left] {$0$};
    [. {\small   $\neg Q$}
       \edge node[pos=.68,auto=right] {$1$};
    [.{\small s.c. $\lor$} 
   ]
    \edge node[auto=left] {$0$};
    [. {\small   s.c. $\neg$}
    ]
    ]]]]
\end{tikzpicture}}
  \] \caption{Examples of strategies for finding and forcing.}\label{fig:strategies for disj/implic}
\end{figure}

\begin{example}
    [Forcing and finding]
    \label{ex:forcing-finding}
From $\{Q_0 \lor Q_1 \mapsto 1 \}$ we can \emph{find} $Q_i\mapsto 1$ in constantly many rounds by:
\begin{itemize}
    \item ask $Q_0$; if $Q_0\mapsto 1$ we are done; else,
    \item ask $Q_1$; if $Q_1\mapsto 1$ we are done; else,
    \item we have a simple contradiction $\{Q_0\lor Q_1 \mapsto 1, Q_0 \mapsto 0, Q_1 \mapsto 0\}$.
\end{itemize}
This (partial) strategy is visualised in \cref{fig:strategies for disj/implic}, left.

Also, writing $Q\cimp R$ as an abbreviation for $ \cnot Q\lor R$, we can \emph{force} $Q\mapsto 1 $ and $R\mapsto 0$ from $\{Q \cimp R \mapsto 0\}$ in constantly many rounds:
\begin{itemize}
    \item ask $Q$ and $R$; if both $Q\mapsto 1$ and $R \mapsto 0$ we are done; else,
    \item if $R \mapsto 1$ we have a simple contradiction against $Q \cimp R \mapsto 0$; else,
    \item we have $Q \mapsto 0$, so ask $\cnot Q$;
    \begin{itemize}
        \item if $\cnot Q\mapsto 1$ we have a simple contradiction against $Q \cimp R \mapsto 0$; else,
        \item if $\cnot Q \mapsto 0$ we have a simple contradiction against $Q \mapsto 0$.
    \end{itemize}
\end{itemize}
This (partial) strategy is visualised in \cref{fig:strategies for disj/implic}, right.
\end{example}

\begin{figure}[t]
    $$\hspace{-18pt}\raisebox{-.008\height}
{\begin{tikzpicture}[scale=.9,grow=up,every tree node/.style={draw=none},
   level distance=1.38cm,sibling distance=.88cm,
   edge from parent path={(\tikzparentnode) -- (\tikzchildnode)}]
   
\Tree
[.{\small $\{\bigwedge \Gamma \mapsto 0\}$} 
\edge[dashed] node[right] {};
[.{\small $\bigwedge \Gamma_0$}
    \edge node[pos=.68,auto=right] {$1$};
    [. {\small $\bigwedge \Gamma_1$}
      \edge node[pos=.68,auto=right] {$1$};
    [.{\small s.c. $\land$} 
   ]
    \edge node[auto=left] {$0$};
    [. {\small  $\IH$}
    ]]
    \edge node[auto=left] {$0$};
    [. {$\IH$}
     ]]]
    \end{tikzpicture}\label{fig:strategy1}}\quad\quad
    {\begin{tikzpicture}[scale=.9,grow=up,every tree node/.style={draw=none},
   level distance=1.38cm,sibling distance=.88cm,
   edge from parent path={(\tikzparentnode) -- (\tikzchildnode)}]
   
\Tree
[.{\small $\{\bigwedge \Gamma \mapsto 1\}$} 
\edge[dashed] node[right] {};
[.{\small $\bigwedge \Gamma_0$}
    \edge node[pos=.68,auto=right] {$1$};
    [.{\small $\IH$} 
   ]
    \edge node[auto=left] {$0$};
    [. {\small $\bigwedge \Gamma_1$}
      \edge node[pos=.68,auto=right] {$1$};
    [.{\small s.c. $\land$} 
   ]
    \edge node[auto=left] {$0$};
    [. {\small  s.c. $\land$}
    ]]]]
\end{tikzpicture}}\label{fig:strategy2}$$
    \caption{$O(\log n ) $ round strategies for forcing and finding in long conjunctions.}
    \label{fig:forcing-finding-long-conjunctions}
\end{figure}

Let us now see a more interesting, and useful, example.
For a list of queries $\Gamma$ we write $\bigwedge \Gamma $ for the conjunction of its members, bracketed as a (nearly) balanced binary tree.
Formally, if $\Gamma = Q_1, \dots, Q_k$ (and $k\geq 2$) then: 
$$\bigwedge \Gamma \dfn \bigwedge\limits_{i=1}^{\lfloor k/2 \rfloor} Q_i \ \cand \ \bigwedge\limits_{i=\lfloor k/2 \rfloor + 1}^k Q_i$$
Similarly for long disjunctions, $\bigvee \Gamma$.
We have:

\begin{example}
\label{finding-forcing-long-combinations}
    Let $\Gamma$ be a nonempty list of queries of length $n$.
    From $\{\bigwedge \Gamma \mapsto b\}$ we can:
    \begin{itemize}
        \item ($b=0$) find $Q\in \Gamma$ with $Q\mapsto 0$ in $O(\log n)$ rounds.
        \item ($b=1$) force $Q\mapsto 1 $ for any $Q \in \Gamma$ in $O(\log n)$ rounds.
    \end{itemize}
    Dually from $\{\bigvee \Gamma\mapsto b\}$ we can:
    \begin{itemize}
        \item ($b=0$) force $Q \mapsto 0 $ for any $Q \in \Gamma$ in $O(\log n ) $ rounds.
        \item ($b=1$) find $Q\in \Gamma$ with $Q\mapsto 1$ in $O(\log n ) $ rounds.
    \end{itemize}

\noindent 
To prove this we proceed by divide-and-conquer induction on $n$, focussing only on the first two items (the other two follow dually). 
The base case, when $n=1$, is immediate.
Otherwise, for $n>1$, let $\Gamma_0$ and $\Gamma_1$ be (roughly) the first and second halves of $\Gamma$, so that $\bigwedge \Gamma = \bigwedge \Gamma_0 \land \bigwedge \Gamma_1$.

\noindent 
For the case $b=0$:

\begin{itemize}
    \item ask $\bigwedge \Gamma_0$ and $\bigwedge \Gamma_1$; 
    \item if $\bigwedge \Gamma_0 \mapsto 0$  find $Q \in \Gamma_0$ with $Q \mapsto 0$, by inductive hypothesis; else,
    \item if $\bigwedge \Gamma_1 \mapsto 0$ find $Q \in \Gamma_1$ with $Q \mapsto 0$, by inductive hypothesis; else,
    \item we have a simple contradiction $\{\bigwedge \Gamma_0 \mapsto 1, \bigwedge \Gamma_1 \mapsto 1, \bigwedge \Gamma \mapsto 0\}$.
\end{itemize}
This is visualised in \cref{fig:forcing-finding-long-conjunctions}, left.

\noindent 
For the case $b=1$, let $Q \in \Gamma_i$:

\begin{itemize}
    \item ask $\bigwedge \Gamma_i$;
    \item if $\bigwedge \Gamma_i \mapsto 1$ force $Q \mapsto 1$ by inductive hypothesis; else,
    \item 
    we have a simple contradiction $\{\bigwedge \Gamma_i \mapsto 0, \bigwedge \Gamma \mapsto 1\}$.
\end{itemize}
This is visualised in \cref{fig:forcing-finding-long-conjunctions}, right.
\end{example}

In the sequel we shall implicitly use similar techniques when constructing strategies.


\subsection{From proofs to strategies.}

Thanks to the presence of Boolean combinations, \cref{games-simulate-proofs} follows by essentially the same proof structure as Pudl\'ak and Buss in \cite{PudBuss}.
The idea is also similar to the translation of Frege proofs into balanced tree-form \cite{krajivcek1994lower}: each sequent is construed as a single query by Boolean combinations, and the winning strategy searches for the first false sequent by a divide-and-conquer on conjunctions of sequents.
This must contradict soundness of some rule, which is again established in logarithmically many rounds.
Thus we first need the following intermediate result:
\begin{lemma}
    [Local soundness]
    \label{local-soundness-strategy}
    Fix an inference step as follows, where $|\Sigma|, |\Pi|, |\Sigma_i|, |\Pi_i| \leq n$:
    \[
\vliiinf{\infrule}{}{\Sigma \seqar \Pi}{\Sigma_{1}\seqar \Pi_{1}}{\cdots }{\Sigma_{k} \seqar \Pi_{k}}
    \]
        From $\{
    \bigwedge \Sigma \cimp \bigvee \Pi \mapsto 0,
    \bigwedge \Sigma_1 \cimp \bigvee \Pi_1 \mapsto 1,
    \dots,
    \bigwedge \Sigma_k \cimp \bigvee \Pi_k \mapsto 1
    \}$
%
    there is a strategy winning in $O(\log n) $ rounds.
\end{lemma}
\begin{proof}
First note that by \cref{ex:forcing-finding} we can: 
\begin{itemize}[align=left] 
    \item[($\star$)] force $\bigwedge \Sigma \mapsto 1$ and force $\bigvee \Pi \mapsto 0$ in constantly many rounds. 
\end{itemize}
Now, almost every case is the same as for the usual sequent calculus $\lk$, equivalently Frege systems (see, e.g., \cite{PudBuss}). The remaining cases are the decision rules of $\elnndt$:
\begin{itemize}
    \item If $\infrule$ is $\vliinf{\lefrul p}{}{\Gamma, \dec A p B \seqar                 \Delta}{\Gamma, A \seqar                 \Delta, p}{\Gamma, p, B \seqar                 \Delta}$ then:
    \begin{itemize}
        \item find $C_0 \in \Gamma,A$ with $C_0\mapsto 0$ or $D_0 \in \Delta,p$ with $D_0\mapsto 1$ in $O(\log n) $ rounds, similarly to \cref{finding-forcing-long-combinations};
        \item find $C_1 \in \Gamma,p,B$ with $C_1 \mapsto 0$ or $D_1 \in \Delta$ with $D_1 \mapsto 1$ in $O(\log n)$ rounds, similarly to \cref{finding-forcing-long-combinations};
        \item if we found some $C_i \in \Gamma$, from $(\star)$ force $C_i \mapsto 1$ in $O(\log n)$ rounds for a contradiction;
        \item if we found some $D_j \in \Delta$, from $(\star)$ force $D_j \mapsto 0 $ in $O(\log n)$ rounds for a contradiction.
    \end{itemize}
    This leaves the possibilities that we found (i) $ A\mapsto 0 $ or $p \mapsto 1$; and, (ii) $ p \mapsto 0 $ or $B \mapsto 0$:
    \begin{itemize}
        \item if $p \mapsto 1 $ and $p\mapsto 0$ then we have immediately a simple contradiction;
        \item in every other case, we have a decision contradiction against $\dec A p B \mapsto 1$.
    \end{itemize}

    \item If $\infrule $ is $\vliinf{\rigrul p}{}{\Gamma \seqar                 \Delta, \dec A p B}{\Gamma \seqar                 \Delta, A, p}{\Gamma, p \seqar                 \Delta , B} $ then:
    \begin{itemize}
        \item find $C_0 \in \Gamma$ with $C_0 \mapsto 0$ or $D_0 \in \Delta A,p$ with $D_0 \mapsto 1$ in $O(\log n)$ rounds, similarly to \cref{finding-forcing-long-combinations};
        \item find $C_1 \in \Gamma,p$ with $C_1 \mapsto 0$ or $D_1 \in \Delta,B$ with $D_1 \mapsto 1$ in $O(\log n)$ rounds, similarly to \cref{finding-forcing-long-combinations};
        \item if we found some $C_i \in \Gamma$, from $(\star)$ force $C_i \mapsto 1$ in $O(\log n)$ rounds for a contradiction;
        \item if we found some $D_j \in \Delta $, from $(\star)$ force $D_j \mapsto 0$ in $O(\log n)$ rounds for a contradiction.
    \end{itemize}
    This leaves the possibilities that we found (i) $A \mapsto 1$ or $p \mapsto 1$ ; and, (ii) $p \mapsto 0$ or $B\mapsto 1$:
    \begin{itemize}
        \item first, from $(\star)$ force $\dec A p B \mapsto 0$ in $O(\log n)$ rounds;
        \item if $p\mapsto 1 $ and $p\mapsto 0$ we immediately have a simple contradiction;
        \item in every other case, we have a decision contradiction against $\dec A p B \mapsto 0$. \qedhere
    \end{itemize}
\end{itemize}
\end{proof}

\noindent 
We can now finally give our translation from proofs to strategies:

\begin{proof}
    [Proof of \Cref{games-simulate-proofs}] 
Fix a proof $\pi:(\Gamma_i \seqar \Delta_i)_{i<n}$ over $\E = \{e_i \extiff E_i\}_{i<k}$ of size $N$ and initial state $S = \{\Gamma_{n-1} \mapsto 1, \Delta_{n-1} \mapsto 0 \}$.
%
Let us employ the following abbreviations:
\begin{itemize}
    \item $L_i \dfn \bigwedge \Gamma_i$, for $i<n$; 
    \item $R_i \dfn \bigvee \Delta_i$, for $i<n$;
    \item $Q_i \dfn L_i \cimp R_i$, for $i<n$; 
    \item $Q^i \dfn \bigwedge\limits_{j<i} Q_i$, for $i\leq n$.
\end{itemize}
First note that, by \cref{ex:forcing-finding,finding-forcing-long-combinations}, from $S$ we can force in $O(\log N)$ rounds $L_{n-1} \mapsto 1 $ and $R_{n-1} \mapsto 0 $ and then also $Q_{n-1} \mapsto 0$, and so finally $Q^n \mapsto 0$.
Now we construct a winning strategy from $Q^{n}\mapsto 0$ as follows:
\begin{itemize}
    \item Find the least $i<n$ such that $Q_i \mapsto 0$, by a divide-and-conquer strategy: keep asking $Q^j$ for $j$ (roughly) halfway between the greatest $ l$ with $Q^l\mapsto 1$ (starting with $l=0$) and least $u$ with $Q^u\mapsto 0$ (starting with $u=n$), until $u-l=1$. Set $i=u$ at this point, so that we have $Q_i \mapsto 0$ but $Q^i \mapsto 1$.
    \item If $\Gamma_i \seqar \Delta_i$ is an extension axiom $e_j \seqar E_j$ or $E_j \seqar e_j$, then we can force a contradiction for extension in constantly many rounds.
    \item If $\Gamma_i\seqar \Delta_i$ is the conclusion of an inference step with premisses (among) $\Gamma_j \seqar \Delta_j$ and $\Gamma_k \seqar \Delta_k$, then $j,k<i$ so we can force $Q_j \mapsto 1$ and $Q_k \mapsto 1$ from $Q^i \mapsto 1$ in $O(\log N)$ rounds by \cref{finding-forcing-long-combinations}.
    Now from $\{Q_j \mapsto 1, Q_k \mapsto 1 , Q_i \mapsto 0\}$ we can win in $O(\log N)$ rounds by \cref{local-soundness-strategy}. \qedhere
\end{itemize}
\end{proof}
 
\subsection{Strategies to proofs: the deterministic case}
\label{sec:strats-to-proofs-deterministic}
In the case of $\eldt$ and $\dbl$, we can readily translate strategies back to proofs thanks to determinism, morally since (non-uniform) logspace is easily closed under Boolean combinations.

\begin{theorem}
\label{thm:eldt-psim-db}
   If $\dbl$ has a size $N$ proof of an \exdt\ sequent $\Gamma \seqar \Delta$ over $\E$ then $\eldt$ has a $\poly (N)$ size proof of $\Gamma \seqar \Delta$ over some $\E' \supseteq \E$.
\end{theorem}
\begin{corollary}
    $\eldt$ polynomially simulates $\dbl$, over \dt\ sequents.
\end{corollary}

\noindent
This result serves as a warm-up before our development of the non-deterministic case in the next sections.  
For the proof it is convenient to work with a proof system that natively manipulates the queries of $\dbl$:
\begin{definition}
    [System for \boolexdt]
    Write $\booleldt$ for the extension of $\eldt$ manipulating Boolean combinations of $\exdt$ formulas, i.e.\ the queries $P,Q, \dots$ of $\dbl$, equipped with the usual rules for negation, disjunction, and conjunction:

\begin{equation}
    \label{eq:pos-bool-rules}
      \begin{array}{ccc}
            \vlinf{\lefrul\lnot}{}{\Gamma , \lnot Q \seqar \Delta}{\Gamma \seqar \Delta, Q} &  
            \vlinf{\lefrul \land}{}{\Gamma, P \land Q \seqar \Delta}{\Gamma, P, Q \seqar \Delta} &
            \vliinf{\lefrul \lor}{}{\Gamma, P \lor Q \seqar \Delta}{\Gamma, P \seqar \Delta}{\Gamma, Q\seqar \Delta}
            \\
             \vlinf{\rigrul\lnot}{}{\Gamma \seqar \Delta, \lnot Q}{\Gamma , Q \seqar \Delta} & 
             \vlinf{\rigrul\lor}{}{\Gamma \seqar \Delta, P\lor Q}{\Gamma \seqar \Delta, P,Q} &
             \vliinf{\rigrul\land}{}{\Gamma \seqar \Delta, P\land Q}{\Gamma \seqar \Delta, P}{\Gamma \seqar \Delta, Q}
        \end{array}
\end{equation}
\end{definition}

By the Boolean constructions of \cite[Section 5]{DBLP:conf/csl/BussDasKnop} (see also \cite{BussDasKnop19:preprint}) we have:
\begin{proposition}
\label{eldt-psim-booleldt}
If $\booleldt$ has a size $N$ proof of an \exdt\ sequent $\Gamma \seqar \Delta$ over $\E$ then $\eldt$ has a $\poly (N)$ size proof of $\Gamma \seqar \Delta$ over some $\E' \supseteq \E$.
\end{proposition}
\noindent 
The branching program construction for Boolean combinations is well-known, and the proof theoretic development of the result is similar to that of the `positive decisions' we shall work with later in \cref{sec:pos-decs}. 
Let us point out, in particular, that negation of \emph{deterministic} branching programs is simple, by the recurrence: 
$\lnot (\dec A p B) \iff \dec {\lnot A} p {\lnot B}$.
Naturally expressing negation is precisely the difficulty in the non-deterministic case.

In any case, from \cref{eldt-psim-booleldt} above, we can readily establish the polynomial simulation of $\dbl$ by $\eldt$:

\begin{proof}
    [Proof of \cref{thm:eldt-psim-db}]
    By \cref{eldt-psim-booleldt} above, it suffices to work with $\booleldt$ instead of $\eldt$. 
    We proceed by induction on the structure of a $\dbl $ strategy from initial state $\{\Gamma\mapsto 1, \Delta \mapsto 0\}$, constructing an $\booleldt$ proof of $\Gamma \seqar \Delta$ over $\E$.
    The construction is straightforward, similar to that of \cite{PudBuss}, simulating each query by a cut maintaining the invariant that queries answered $1$ are on the LHS of the sequent and queries answered $0$ are on the RHS of the sequent. 
    Formally, if $\sigma$ begins by querying some $P$, with substrategies $\sigma_0,\sigma_1$ for the answers $0,1$ respectively, we construct the following $\booleldt$ proof,
    \[
    \vlderivation{
    \vliin{\cut}{}{\Gamma \seqar \Delta}{
        \vltr{\IH}{\Gamma \seqar \Delta, P}{\vlhy{\ }}{\vlhy{\quad }}{\vlhy{\ }}
    }{
        \vltr{\IH}{\Gamma, P \seqar \Delta }{\vlhy{\ }}{\vlhy{\quad }}{\vlhy{\ }}
    }
    }
    \]
    where the subproofs marked $\IH$ are obtained by the inductive hypothesis, with $\sigma_0,\sigma_1$ construed as winning strategies from the corresponding intial states.
    Thus it remains to derive the base cases of simple contradictions by polynomial size proofs:\footnote{We omit consideration of weakenings immediately below an initial sequent in all cases.}
    \begin{itemize}
        \item Boolean contradictions are given by the corresponding initial Boolean rule.
        \item Decision contradictions are given by constant length proofs using just the decision, structural and initial rules.\footnote{In fact, formally, constant size such proofs must exist by the previous completeness result, \cref{soundness-completeness}.}
        For instance $\{\dec A p B \mapsto 1, p \mapsto 0, A \mapsto 0\}$ and $\{\dec A p B \mapsto 0 , p \mapsto 1 , B \mapsto 1\}$ are respectively given by,
        \[
        \vlderivation{
        \vliin{\lefrul p}{}{\dec A p B \seqar p, A}{
            \vliq{\wk}{}{A \seqar p, p , A}{
            \vlin{\id}{}{A \seqar A}{\vlhy{}}
            }
        }{
            \vliq{\wk}{}{p, B \seqar p, A}{
            \vlin{\id}{}{p\seqar p}{\vlhy{}}
            }
        }
        }
        \qquad
        \text{and}
        \qquad
        \vlderivation{
        \vliin{\rigrul p }{}{p, B \seqar \dec A p B}{
            \vliq{\wk}{}{p, B \seqar A,p}{
            \vlin{\id}{}{p\seqar p}{\vlhy{}}
            }
        }{
            \vliq{\wk}{}{p, B, p \seqar B}{
            \vlin{\id}{}{B\seqar B}{\vlhy{}}
            }
        }
        }
        \quad \text{.}
        \]
        \item Boolean connective contradictions are the same as for $\lk$, cf.~\cite{PudBuss}.
        \item Extension contradictions are given by the corresponding extension axioms.
        \item Similarity contradictions are given by \cref{simulation-has-polysize-proofs}. \qedhere
    \end{itemize}
\end{proof}

\begin{remark}
    [On tree-like proofs]
Note that, since strategies are trees, the proof constructed above is almost tree-like, except for the subproofs for similarity contradictions arising from \cref{simulation-has-polysize-proofs}. 
However the simulation of $\booleldt$ by $\eldt$, \cref{eldt-psim-booleldt}, does not preserve this tree-like structure.
\end{remark}

\section{A non-uniform partial formalisation of Immerman-\szel}\label{sec: Imm-szel}
\label{sec:immszel}

In order to translate strategies to proofs, we need a way to simulate the \emph{negation} of branching programs.
For this, in the non-deterministic case, we need a non-uniform formalisation of the Immerman-\szel\ theorem, $\coNL=\NL$.
Throughout this section we focus only on $\elndt$, not $\eldt$.

\subsection{Working with positive decisions}
\label{sec:pos-decs}

It will be convenient in our construction to stitch together branching programs by `positive' decisions.
This sort of decision induces what are often called `monotone' or `positive' (non-deterministic) branching programs (PBPs): NBPs where, for every $0$-edge, there is a parallel $1$-edge \cite{GrigniSnipser,grigni1991structure}, for instance as in~\cref{example-2-4-thresh} earlier. 
We investigated the proof complexity of PBPs in earlier work \cite{Das-Del,DasDel25:pos-bps-journal} essentially by restricting decision formulas to have the form
 $\posdec A p B$, hence comprising a bona fide sublanguage of eNDT formulas.
Here it will be convenient to `bootstrap' our language with the capacity to make (certain) positive decisions on \emph{compound} formulas (over a fixed set of extension variables), and even recursively so.

\begin{definition}
    [Positive decisions]
    \label{pos-decision-programs}
    Let $\E = \{e_i \extiff E_i\}_{i<n}$ be a set of extension axioms and $\vec e=\{e_i\}_{i<n}$.
    We introduce fresh extension variables by the closure property:
    \begin{itemize}
        \item if $B$ is a formula over $\vec e$, and $A,C$ are formulas, then $\posdec A B C$ is an extension variable. 
    \end{itemize}
    Over this language 
    we define $\pos \E$ to be the smallest extension of $\E$ by all extension axioms of the form:
    \[
    \begin{array}{r@{\ \extiff \ }l}
         \posdec A 0 C & A \\
         \posdec A 1 C & A \lor C \\
         \posdec A{(B_0 \lor B_1) }C & (\posdec A{B_0}C )\lor( \posdec A{B_1}C) \\
         \posdec A{(\dec {B_0}p{B_1})}C & \dec{(\posdec A{B_0}C)}p{(\posdec A{B_1}C)} \\
         \posdec A{e_i} C & \posdec A{E_i}C 
    \end{array}
    \]
\end{definition}

The notation we have used is intentionally suggestive, under the interpretation by $\E^+$, perhaps at the danger of ambiguity: we insist, for foundational reasons, that $\posdec A B C$ are \emph{not} compound formulas, but rather bona fide extension variables.
Note that the formulas $A$ and $C$ above may contain extension variables besides those in $\vec e$, in particular other positive decisions according to the closure property. 
For instance we have extension variables of the form $\posdec A B {(\posdec C {B'} D)}$ and so on. However $B$ and $B'$ here must be over the original set of extension variables $\vec e$.

By the $\E$-induction principle, \cref{A-induction}, we may readily establish the following characteristic truth conditions: 

\begin{lemma}
    [Truth for positive decisions]
    \label{truth-pos-decs}
    Fix a set of extension axioms $\E = \{e_i \extiff E_i\}_{i<n}$. 
    For formulas $B$ over $\vec e$
    there are polynomial size $\elndt$ proofs over $\pos \E$ of:
    \[
    \begin{array}{r@{ \ \seqar \ }l}
        \posdec A B C & A,B \\
        \posdec A B C & A,C 
        \\
        A & \posdec A B C \\
        B,C & \posdec A B C
    \end{array}
    \]
\end{lemma}


\begin{proof}
The argument is a straightforward $\E$-induction on $B$. We will only present proofs for $ B,C \seqar \posdec ABC$ as the other sequents follow similarly.

In the base case, when $B$ is $0$ or $1$, we have:
\[
\vlderivation{
     \vliq{\wk}{}{ 0,C \seqar B}{
       \vlin{\lefrul 0 }{}{0\seqar}{\vlhy{}}
     }
     }
     \qquad
     \vlderivation{
     \vliq{\wk}{}{ 1,C \seqar A,C}{
       \vlin{\id}{}{C\seqar C}{\vlhy{}}
     }
     }
     \]

If $B=D\lor E$ we have,

\[
\vlderivation{\vlin{\pos \E,\rigrul{\lor}}{}{ D\lor E,C \seqar \posdec{A}{(D\lor E)}{ C}}{
     \vliin{\lefrul{\lor}}{}{ D\lor E,C \seqar \posdec{A}{D}{ C}, \posdec{A}{E}{ C}}{
        \vliq{\IH}{}{  D,C\seqar \posdec{A}{D}{ C}}{\vlhy{}{}{}{}}
     }{
        \vliq{\IH}{}{  E,C\seqar \posdec{A}{E}{ C} }{\vlhy{}{}{}{}{}}
     }}
     }
     \]
where the steps marked $\IH$ are obtained by the inductive hypothesis.

    If $B=\dec D p  E$ we have the two proofs:
    \[
    \vlderivation{
        \vliin{\pos \E,\rigrul p}{}{  C,D\seqar p,\posdec {A}{(\dec D pE)}{C}}{\vliq{\wk,\exch}{}{C,D \seqar p,p,\posdec{A}{D}{C}}{
        \vliq{\IH}{}{D,C \seqar \posdec A D C }{\vlhy{}}
        }{}{}{}{}}{\vliq{ \wk}{}{C,D,p \seqar p,\posdec{A}{E}{C}}{\vlin{\id}{}{p\seqar p}{\vlhy{}}}{}{}{}{}}
     }
    \]
    \[
    \vlderivation{
      {
        \vliin{\pos \E,\rigrul p
        }{}{  C,p,E\seqar \posdec {A}{(\dec D pE)}{C}}{\vliq{\wk}{}{C,p,E \seqar p,\posdec {A}{D}{C}}{\vlin{\id}{}{p\seqar p}{\vlhy{}}}{}{}{}{}}{\vliq{\wk,\exch}{}{C,p,E \seqar \posdec {A}{E}{C}}{
            \vliq{\IH}{}{E,C \seqar \posdec A E C}{\vlhy{}}
        }{}{}{}{}}}}
    \]
    which we can combine to obtain, as required:

\[
\vlderivation{
     \vliin{\lefrul p}{}{ \dec D pE,C \seqar \posdec {A}{(\dec D pE)}{C}}{
     \vlhy{  C,D\seqar p,\posdec {A}{(\dec D pE)}{C}}
     }{
     \vlhy{  C,p,E\seqar \posdec {A}{(\dec D pE)}{C}}
     } }
\]
     

        If $B$ is an extension variable $\ee{i}{}{}$ with extension axiom $\ee{i}{}{}\extiff E_i$, we just call the inductive hypothesis for  $\posdec {A}{E_i}{C}$.
\end{proof}

\begin{remark}
[Positive decision rules]
    \label{rem:posdec-rules}
    Note that, from the Truth Conditions above, we can immediately derive the `positive decision' rules as in \cite{Das-Del,DasDel25:pos-bps-journal}, that we may use in the sequel:
\begin{equation}
    \label{eq:posdec-rules}
    \vliinf{\lefrul{B^+}}{}{\Gamma, \posdec A B C \seqar \Delta}{\Gamma ,A\seqar \Delta}{\Gamma, B,C \seqar \Delta}
\qquad
\vliinf{\rigrul{B^+}}{}{\Gamma \seqar \Delta, \posdec A B C}{\Gamma\seqar \Delta, A,B}{\Gamma \seqar \Delta, A , C}
\end{equation}
\end{remark}

\subsection{Positive counting programs and their properties}
\label{sec:pos-counting-progs}
\label{thresholds-over-programs}
For the remainder of this section let us fix $\vec B = B_0, \dots, B_{N-1}$ over a set of extension axioms $\mathcal B$.
Introduce new (temporarily uninterpreted) extension variables $\thresh {\Bj j}  k $ for each $\Bj j  = B_j, \dots, B_{N-1}$, for $j\leq N$ (so $\Bj N$ is just the empty list $\epsilon$) and $k \in \mathbb Z$.
After generating $\pos{\mathcal B}$ as in~\cref{pos-decision-programs} above, we define $\thrextaxs {\vec B}{\mathcal B}$ to be the smallest extension of $\pos{\mathcal B}$ by all extension axioms of the form:
\[
\begin{array}{r@{\ \extiff \ }ll}
     \thresh {\epsilon  }  k & 1 & \text{for $k\leq 0$} \\
     \thresh \epsilon {k} & 0 & \text{for $k> 0$}
     \\
     \thr { \Bj j } {k} & 
     \posdec {\thresh{\Bj{j+1}}{k} } {B_j} { \thresh{\Bj {j+1}} {k-1} } & \text{for $j<N$} 
\end{array}
\]
It is not hard to see, by induction on $|\Bj j|$, that $\thresh {\Bj j } k$ is true over $\thrextaxs \B \BB$ just when at least $k$ of $\Bj j$ are true.
Note that we have extension variables with negative subscripts too, like $\thr \epsilon {-1}$, which we have set to $1$. 
It is convenient to admit these so that definitions and proofs by induction on the length of the superscript are more uniform (like the final case above).

\begin{remark}
    Notice that the thresholds programs above are defined slightly differently from the ones from the earlier \cref{example-2-4-thresh}, in that the leaves left of the $1$ sink in \cref{fig:thresh-nbp-4-2} are set to $1$ rather than $0$ by the above formulation.
    By monotonicity of the program, this makes no difference to the semantics, but the current formulation streamlines the exposition henceforth.
\end{remark}

Similar `threshold' programs were introduced in \cite{Das-Del,DasDel25:pos-bps-journal} where several basic properties were established in order to carry out counting arguments.
As our formulation is slightly different, and for self-containment, we reproduce two necessary lemmata here, with reference to the analogous results from previous work.

\renewcommand{\storageone}{\text{\cref{truth-pos-decs}}}

\begin{lemma}
[Monotonicity of thresholds, cf.~{\cite[Proposition 4.6]{DasDel25:pos-bps-journal}}]
 \label{monotonicity-of-threshold-subscripts}
 There are polynomial-size $\elndt$ proofs over $\thrextaxs{\vec B}{\BB}$ of:
 \begin{enumerate}
     \item\label{item:thr-0-true} $\quad \seqar \thr {\Bj j} k$ whenever $k\leq 0$
    \item\label{item:thr-big-false} $\thr{\Bj j}{k} \seqar  \quad $ whenever $k>|\Bj j|$
     \item\label{item:thr-k+1-implies-thr-k} $\thr{\Bj j}{k} \seqar \thr{\Bj j }{k-1}$
 \end{enumerate}
\end{lemma}

\begin{proof}
    All items are proved by induction on $|\Bj j|$. The base cases, when $\Bj j = \epsilon$, are immediate from the extension axioms $\thrextaxs {\vec B} \BB$.
    We derive the inductive steps as follows,
        \[
        \eqref{item:thr-0-true} : \quad
        \vlderivation{
        \vlin{\thrextaxs{\vec B} \BB}{}{\seqar \thresh {\Bj j} k }{
        \vliq{\storageone}{}{\seqar \posdec{ \thresh{\Bj {j+1}}{k} }{B_j}{\thresh{\Bj {j+1}}{k-1}} }{
        \vliq{\IH}{}{\seqar \thresh {\Bj {j+1}} k }{\vlhy{}}
        }
        }
        }
        \qquad
        \eqref{item:thr-big-false} : \quad
        \vlderivation{
        \vlin{\thrextaxs{\vec B} \BB}{}{\thresh{\Bj j}{k} \seqar }{
        \vliiq{\storageone}{}{ \posdec{ \thresh {\Bj {j+1}}{k} }{B_j}{\thresh{\Bj{j+1}}{k-1}} \seqar }{
            \vliq{\IH}{}{\thresh{\Bj{j+1}}{k} \seqar }{\vlhy{}}
        }{
            \vliq{\IH}{}{\thresh{\Bj{j+1}}{k-1} \seqar }{\vlhy{}}
        }
        }
        }
        \]
        \[
\eqref{item:thr-k+1-implies-thr-k} :
\quad
        \vlderivation{
        \vliq{\thrextaxs{\vec B} \BB}{}{\thresh{\Bj j }{k} \seqar \thresh{\Bj j}{k-1} }{
        \vliiq{\storageone}{}{ \posdec{\thresh{\Bj {j+1}}{k}}{B_j}{\thresh{\Bj{j+1}}{k-1}} \seqar \posdec{\thresh{\Bj {j+1}}{k-1}}{B_j}{\thresh{\Bj{j+1}}{k-2}} }{
            \vliq{\IH}{}{\thresh{\Bj{j+1}}{k} \seqar \thresh{\Bj {j+1}}{k-1}}{\vlhy{}}
        }{
            \vlin{\id}{}{\thresh{\Bj{j+1}}{k-1} \seqar \thresh{\Bj{j+1}}{k-1} }{\vlhy{}}
        }
        }
        }
        \]
        where the steps marked $\IH$ are obtained by the inductive hypothesis.
\end{proof}

\begin{lemma}
[Truth for thresholds, cf.~{\cite[Lemma 5.4]{DasDel25:pos-bps-journal}}]
\label{thresh-truth}
There are polynomial-size $\elndt$ proofs over $\thrextaxs{\vec B} \BB$, of:
    \begin{enumerate}
        \item\label{item:thresh-truth-Bj-true} $B_j, \thresh {\Bj {j+1}} k \seqar \thresh {\Bj j} {k+1}$
        \item\label{item:thresh-truth-j-dec} $\thresh {\Bj {j+1}} k \seqar \thresh {\Bj j} k$
        \item\label{item:thresh-truth-Bj-false} $\thresh{\Bj j} k \seqar \thresh {\Bj {j+1}} k , B_j$
        \item\label{item:thresh-truth-j-inc} $\thresh{\Bj j } k \seqar \thresh{\Bj {j+1}} {k-1}$
    \end{enumerate}
\end{lemma}
\renewcommand{\storagetwo}{\text{\cref{truth-pos-decs} \& \cref{monotonicity-of-threshold-subscripts}}}

\begin{proof}
Almost all cases follow by simply unfolding a threshold extension variable by $\thrextaxs{\vec B} \BB$ and applying the truth conditions for positive decisions:
    \[
\eqref{item:thresh-truth-Bj-true}: \ 
    \vlderivation{
    \vlin{\thrextaxs {\vec B} \BB}{}{B_j, \thresh {\Bj {j+1}} k \seqar \thresh {\Bj j } {k+1}}{
    \vliq{}{}{B_j, \thresh{\Bj{j+1}}{k} \seqar \posdec{\thresh {\Bj {j+1}  }k} {B_j}{\thresh{{\Bj {j+1}}}{k} } }{\vlhy{\storageone}}
    }
    }
    \qquad
\eqref{item:thresh-truth-j-dec}: \ 
    \vlderivation{
    \vlin{\thrextaxs {\vec B} \BB}{}{\thresh {\Bj {j+1}} k \seqar \thresh {\Bj j} k }{
    \vliq{}{}{\thresh {\Bj {j+1}} k \seqar \posdec {\thresh {\Bj {j+1}} k}{B_j}{\thresh {\Bj {j+1}} {k-1}} }{\vlhy{\storageone}}
    }
    }
    \]
    \[
    \eqref{item:thresh-truth-Bj-false}:\ 
    \vlderivation{
    \vlin{\thrextaxs {\vec B} \BB}{}{\thresh {\Bj j} k \seqar \thresh {\Bj {j+1}} k , B_j }{
    \vliq{}{}
    {\posdec{\thresh {\Bj {j+1}} k }{B_j}{\thresh {\Bj {j+1}} {k-1} } \seqar \thresh {\Bj {j+1}} k , B_j }
    {\vlhy{\storageone}}
    }
    }
    \qquad
\eqref{item:thresh-truth-j-inc}:\ 
    \vlderivation{
    \vlin{\thrextaxs {\vec B} \BB }{}{\thresh{\Bj j} k \seqar \thresh{\Bj{j+1}}{k-1}}{
    \vliq{}{}{\posdec{\thresh{\Bj{j+1}}{k}}{B_j}{\thresh{\Bj{j+1}}{k-1}} \seqar \thresh{\Bj{j+1}}{k-1} }{
    \vlhy{\storagetwo}
    }
    }
    }
    \]
    In the final case, \eqref{item:thresh-truth-j-inc}, we are also using the monotonicity property $\thresh {\Bj{j+1}} k \seqar \thresh{\Bj{j+1}}{k-1}$, by \cref{monotonicity-of-threshold-subscripts}.\eqref{item:thr-k+1-implies-thr-k}.
\end{proof}

\subsection{Decider construction}
\label{sec:decider-construction}
We shall employ the counting mechanisms from \cref{thresholds-over-programs} to define a construction that will compute a decision on \exndt\ formulas relative to a fixed number of true formulas, over appropriate extension axioms. 
For this we shall employ ideas from the Immerman-\szel\ Theorem, $\coNL=\NL$, in a nonuniform setting.

Recall that, from the beginning of \cref{sec:pos-counting-progs}, we have fixed  formulas $\vec B = B_0, \dots, B_{N-1}$ 
over a set of extension axioms $\mathcal{ B}$, 
and we have generated $\pos{{\mathcal{ B}}}$ and 
$ \thrextaxs{\vec B}{\BB}$.
%
%
The main result of this section is:
\begin{theorem}
[Formalised (partial) Immerman-\szel]
\label{thm:imm-szel}
For $i<N, k \in \mathbb Z $ and formulas $A,C$ there are extension variables $\dec A{B_i^k}C$ 
and a set of extension axioms extending $\thrextaxs{\vec B}{\BB}$ over which there are polynomial-size $\elndt$ proofs of:
\begin{equation}
    \label{eq:decider-truth-conds-relativised}
    \begin{array}{r@{\ \seqar \ }l}
     \gray{\thr{\vec B}k,}\dec A{B_i^k}C & A, B_i \gray{, \thr{\vec B}{k+1}} \\
     \gray{\thr{\vec B}k ,} \dec A{B_i^k}C , B_i & C \gray{, \thr{\vec B}{k+1} }
     \\
     \gray{\thr{\vec B}k  ,} B_i , C & \dec A{B_i^k}C \gray{, \thr{\vec B}{k+1} }\\
     \gray{\thr{\vec B}k  , }A  & B_i, \dec A{B_i^k}C\gray{, \thr{\vec B}{k+1}}
\end{array}
\end{equation}
\end{theorem}

The idea behind the sequents above is that, as long as \emph{exactly} $k$ of $\vec B$ are true, $\decider k A{B_i}C$ accurately computes a decision on $B_i$, returning the value of $A$ if false, and $C$ if true. 
Indeed, ignoring the gray threshold formulas, the sequents above are just the truth conditions for a decision `$\dec A {B_i} C$'.
Note here that having $\thr{\vec B}{k+1}$ on the RHS is semantically equivalent to having its negation on the LHS, and of course $\thr{\vec B}k \land \lnot \thr{\vec B}{k+1}$ holds just when exactly $k$ of $\vec B$ are true.
The notation $\dec A{B_i^k}C$ we have used is suggestive but we emphasise that these expressions are \emph{not} compound formulas, but rather bona fide extension variables.
We will call $\decider k A{B_i}C$ the 
\emphasis{$k$-decider} of $B_i$ (for $A,C$, over $\vec B$).

\begin{figure}[t]
\[
\small
\begin{tikzpicture}[scale=1.2, auto,swap]     
\foreach \pos/\name/\disp in {
  {(0,4)/1/$ {\blue{\text{{$B_0$}}}}$}, 
    {(0.5,4)/1.5/$\blue{(0,0)}$}, 
  {(-1,3)/2/$ {\blue{\text{{$B_1$}}}}$},
    {(-.5,3)/2.5/$\blue{(0,0)}$},
  {(1,3)/3/$ {\blue{\text{{$B_1$}}}}$}, 
    {(1.5,3)/3.5/$\blue{(1,0)}$},
  {(-2,2)/4/$ {\blue{\text{{$B_2$}}}}$}, 
    {(-1.5,2)/4.5/$\blue{(0,0)}$},
  {(0,2)/5/$ {\blue{\text{{$B_2$}}}}$},
    {(0.5,2)/5.5/$\blue{(1,0)}$},
  {(2,2)/6/$ {\blue{\text{{$B_2$}}}}$}, 
    {(2.5,2)/6.5/$\blue{(2,0)}$},
  {(0,1.53)/7/$\vdots$}, 
  {(-2,1.53)/8/\vdots}, 
  {(2,1.53)/9/\vdots},
  {(-2.8,-.3)/10/},
  {(2.8,-.3)/11/},
  {(-3.8,-1.62)/12/$ {\red{\text{{$B_{N-1}$}}}}$},
    {(-2.88,-1.62)/12.5/$\red{(k-1,0)}$},
  {(3.8,-1.62)/13/$ {\green{\text{{$B_{N-1}$}}}}$},
      {(4.5,-1.62)/13.5/$\green{(k,1)}$},
  {(-1.2,-1.62)/14/$ {\red{\text{{$B_{N-1}$}}}}$},
      {(-.53,-1.62)/14.5/$\red{(k,0)}$},
  {(1.4,-1.62)/29/$ {\green{\text{{$B_{N-1}$}}}}$},
    {(2.35,-1.62)/29.5/$\green{(k-1,1)}$},
   {(0.08,-1.68)/30/$\hdots$},
     {(0.08,-2.38)/31/$\hdots$},
       {(-5.28,-2.38)/32/$\hdots$},
         {(5.28,-2.38)/33/$\hdots$},
  {(-4.6,-2.8)/15/$\textcolor{orange}{0}$},
  {(-2.4,-2.8)/16/$ \red A$},
  {(-2.3,-2.8)/17/},
  {(-0.6,-2.8)/18/$\textcolor{orange}{0}$},
  {(+2.82,-2.8)/19/$\green C$},
  {(+4.6,-2.8)/20/$\textcolor{orange}{0}$},
  {(.8,-2.8)/34/$\textcolor{orange}{0}$},
  {(-2,-1.18)/21/$\hdots$},
  {(+2,-1.18)/22/$\hdots$},
   {(0,1.03)/23/$ {\blue{\text{{$B_i$}}}}$},
    {(.48,1.03)/23.5/$\blue{(c,0)}$},
   {(-0.88,-.03)/24/$ {\red{\text{{$B_{i+1}$}}}}$},
    {(-.288,-.03)/24.5/$\red{(c,0)}$},
   {(+0.88,-.03)/25/$ {\green{\text{{$B_{i+1}$}}}}$},
     {(1.75,-.03)/26.5/$\green{(c+1,1)}$},
   {(-.88,-.48)/26/$\vdots$},
   {(.88,-.48)/27/$\vdots$},
   {(-5.28,-1.68)/28/$\hdots$},
   {(5.28,-1.68)/28/$\hdots$}}
\node[minimum size=20pt,inner sep=0pt] (\name) at \pos {\disp};

  \draw [->][thin]
    (1) [out=180, in=100] to  (2);
    \draw [->][thick,dotted](1) to (2);
    \draw [->][thin](1) to (3);
    \draw [->][thin]
    (2) [out=180, in=100] to  (4);
    \draw [->][thick,dotted](2) to (4);
    \draw [->][thin](2) to (5);
     \draw [->][thin]
    (3) [out=180, in=100] to  (5);
    \draw [->][thick,dotted](3) to (5);
    \draw [->][thin](3) to (6);
    \draw [->][thin]
    (10) [out=180, in=100] to  (12);
    \draw [->][thick,dotted](10) to (12);
    \draw [->][thin](11) to (13);
    \draw [->][thin]
    (12) [out=180, in=100] to  (15);
     \draw [->][thick,dotted](12) to (15);
     \draw [->][thin](12) to (16);
     \draw [->][thin]
    (14) [out=180, in=100] to  (17);
     \draw [->][thick,dotted](14) to (17);
     \draw [->][thin](14) to (18);
     \draw [->][thin](13) to (20);
     \draw [->][thin]
    (13) [out=180, in=100] to  (19);
     \draw [->][thick,dotted](13) to (19);
     \draw [->][thin]
    (23) [out=180, in=100] to  (24);
    \draw [->][thick,dotted](23) to (24);
    \draw [->][thin](23) to (25);
    
    \draw [->][thin]
    (29) [out=180, in=100] to  (34);
    \draw [->][thick,dotted](29) to (34);
    \draw [->][thin](29) to (19);

\end{tikzpicture}
\]
\caption{The $k$-decider $\dec A{B_i^k} C$, over $\dextaxs \B \BB$, visualised as an NBP that is a positive combination of $\vec B$. $0$-edges are dotted and $1$-edges are solid. 
}\label{decider}
 \end{figure}

To give the intuition before its formal definition, the corresponding NBP that will be computed by $\dec A{B_i^k}C$ is visualised (over its corresponding extension axioms) in \cref{decider} as a NBP, in fact as a positive combination of the programs $\vec B$ we started with.  
Note here that each `node' is actually a copy of an NBP $B_j$, indexed by a pair $(c,b)$. $c$ is a counter that tracks how many true $B_j$s along the path thusfar, while $b$ is a Boolean flag that flips from $0$ to $1$ if $B_i$ is true along the path. 
The blue part of the diagram is all before deciding $B_i$, after which the program splits into two cases, indicated in red ($b=0$) or green ($b=1$), depending on whether the $0$-direction or $1$-direction, respectively, from $B_i$ was followed.
The orange $\color{orange} 0$ leaves correspond to paths where the counter does not match the number of true variables $k$, and so the program returns erroneously.
The extension variables used to describe this program will duly have three indices $c,b,j$.

Let us now define the $k$-decider $\dec A {B_i^k} C$ more formally and prove its relevant properties.
For the remainder of this section we fix $i<N $, $k \in \mathbb Z$ and formulas $A,C$, towards proving \cref{thm:imm-szel}.

\begin{definition}
    [Decision programs]
    \label{dec-ext-dfn}
    We introduce extension variables $\decext j c b$, for $j\leq N, c\leq N, b<2$, and write $\dextaxs \B \BB$ for the smallest extension of $\thrextaxs \B \BB$ by the following extension axioms, given by case analysis on $j\leq N$:
    \begin{itemize}
        \item For $j< N, j\neq i$: 
        $$\decext  j c b \dseqar \posdec {\decext {j+1} c b} {B_j} {\decext {j+1}{c+1}b}$$
        \item For $j=i$:
        \[
        \blue{\decext  i c 0} \dseqar \posdec {\red{\decext  {i+1} c 0}} {B_i} {\green{\decext  {i+1}{c+1}1}}
        \]
        \item For $j=N$:
        \[
        \begin{array}{r@{\ \dseqar\ }ll}
             \red{\decext  N k 0} & \red A & \\
             \green{\decext  N k 1} & \green C & \\
             \decext  N c b & \orange 0 & \text{if $c\neq k$}          
        \end{array}
        \]

        From here we set $\dec A {B_i^k} C \dfn \decext 0 0 0$.
    \end{itemize}
\end{definition}

\begin{remark}
[Relation to \cite{DBLP:journals/jcss/AtseriasGP02} and Immerman-\szel]
The formulas $\dec 1{B_i^k} 0$, by \cref{thm:imm-szel}, will represent the negation of $B_i$ in environments where exactly $k$ of $\vec B$ are true.
Formally, in the terminology of \cite{DBLP:journals/jcss/AtseriasGP02} they are \emph{pseudocomplements} with respect to the `exact-$k$' functions.
However our construction of these pseudocomplements is rather inspired by the inductive counting argument underlying the Immerman-\szel\ theorem:
\begin{itemize}
    \item Each `node' non-deterministically checks whether some $B_i$ is true, i.e.\ whether its source reaches the sink $1$, similarly to the subroutine of Immerman-\szel. 
    The role of the counter $c$ is similar to the local counter of the subroutine, with the computation aborting (the orange $\orange 0$s of \cref{decider}) if the wrong number of $B_j$'s were guessed true.
    \item Instead of a main routine that inductively calculates the number of true $B_j$'s (i.e.\ the number of $B_j$'s whose sources eventually reach $1$), an exact number $k$ of true $ B_j$'s is given to us in the context of the corresponding truth conditions, cf.~\cref{eq:decider-truth-conds-relativised}. 
    The role of the main routine is rather carried out by a counting argument at the level of the proof system, essentially a (non-uniform) induction on $k$.
\end{itemize}
\end{remark}

\subsection{Some intermediate results}







In order to prove \cref{thm:imm-szel}, 
we first establish a series of intermediate lemmata, characterising evaluation at each node of the decider.
Each result will rely on the previous ones, working from the leaves of the $k$-decider upwards, cf.~\cref{decider}.

All proofs in this subsection are in the system $\elndt$, over extension variables $\dextaxs{\vec B} \BB$ (which also includes $\pos {\mathcal B}$ and $\thrextaxs{\vec B}\BB$). 
Recall that we fixed $i<N$, $k\in \mathbb Z$ and formulas $A,C$ in the previous subsection.

Intuitively,
at each node $\decext j c b $, the formulas $B_0, \dots, B_{j-1}$ have already been `evaluated', and the counter $c$ tells us that (at most) $c$ of them have returned true.
How the remainder of the program evaluates now depends on the number $l$ of true formulas among $\Bj j = B_{j}, \dots, B_{N-1} $.
If $l,c$ are too low for the counter to eventually reach $k$ then the $k$-decider will return erroneously (the orange $\orange 0$ leaves in \cref{decider}).
This intuition is formalised by the following result:

\begin{lemma}
\label{decext-char-low-count}
    If $l + c  < k$ there are polynomial-size proofs of $\gray{\thresh{\vec B^j}{l} ,} \decext j c b \seqar \gray{\thresh{\vec B^j}{l+1}}$, for $j\leq N$.
\end{lemma}
\begin{proof}
    By backwards induction on $j\leq N$:
    \begin{itemize}
        \item For the base case, $j=N$, we have $\vec B^N = \epsilon$ so:
        \begin{itemize}
            \item if $c\neq k$ we have $\decext N c b \dseqar 0$ by $\dextaxs{\vec B}\BB $; 
            \item if $c=k$ then $l<0$
        so $\thresh{\epsilon} {l+1} \dseqar 1 $ by $\thrextaxs{\vec B}{\BB}$.
        \end{itemize} 
        In both cases we may conclude by an initial step $\lefrul 0 $ or $\rigrul 1$ and weakenings.
        \item For the inductive step we have the following derivation, conducting a sort of case analysis on $B_j$:
        \[
        \vlderivation{
        \vliin{\cut}{}{\thresh{\vec B^j} l , \decext j c b \seqar \thresh {\vec B^j} {l+1}}{
            \vliq{\dextaxs \B \BB}{}{\thresh{\vec B^j} l , \decext j c b \seqar \thresh {\vec B^j} {l+1}, B_j}{
            \vliq{\IH}{}{\thresh {\vec B^j} {l}, \decext {j+1}c {b} \seqar \thresh {\Bj j} {l+1} }{\vlhy{}}
            }
        }{
        \vliiq{\dextaxs {\B} \BB }{}{B_j, \thresh {\vec B^j} {l} , \decext j c b  \seqar \thresh {\vec B^j} {l+1}} {
            \vliq{\IH}{}{\thresh {\vec B^j} {l}, \decext {j+1}c {b} \seqar \thresh {\Bj j} {l+1} }{\vlhy{}}
        }{
            \vliq{\thrextaxs{\vec B}{\BB}}{}{\thresh {\vec B^j} {l}, B_j, \decext{j+1}{c+1}{b'} \seqar \thresh {\vec B^j} {l+1}}{
            \vliq{\IH}{}{\thresh {\vec B^{j+1}} {l-1}, \decext{j+1}{c+1}{b'} \seqar \thresh {\vec B^{j+1}} {l}}{\vlhy{}}
            }
        }
        }
        }
        \]
        for appropriate $b'$, where the steps marked $\IH$ are obtained by the inductive hypothesis, steps marked $\dextaxs{\vec B}\BB$ are obtained by positive truth conditions from \cref{truth-pos-decs} under the corresponding extension axioms from \cref{dec-ext-dfn}, and steps marked $\thrextaxs {\vec B}\BB $ are obtained by the truth conditions for thresholds from \cref{thresh-truth}. \qedhere
    \end{itemize}
\end{proof}

$\elndt$ proofs in the remainder of this subsection will have a similar structure to the one above, and we shall duly annotate them similarly without further clarification.

Continuing our informal intuition, if $l+c=k$ then the $k$-decider will return correctly (i.e.\ it will correctly decide $ B_i$, returning $A$ if false and $C$ if true).
If $j>i$ then $\decext j c b$ has already evaluated $B_i$ and the latter's truth value is stored as $b$.
Thus such nodes will return $A$ or $C$ depending on the value of $b$.
This intuition is formalised by the following result:

\begin{lemma}
\label{decext-char-after-bool-flag}
    If $l+c= k$ then, for $i<j\leq N$, there are polynomial-size proofs of:
        \begin{enumerate}
            \item\label{item:decext-implies-zerocase} $\gray{\thresh{\vec B^j}{l},} \decext j c 0  \seqar A \gray{,\thresh{\vec B^j}{l+1}}$
            \item\label{item:zerocase-implies-decext} $\gray{\thresh{\vec B^j}{l} ,}A\seqar\decext j c 0 \gray{,\thresh{\vec B^j}{l+1}}$
            \item\label{item:decext-implies-onecase} $\gray{\thresh{\vec B^j}{l} , }\decext j c 1 \seqar C \gray{, \thresh{\vec B^j}{l+1}}$
            \item\label{item:onecase-implies-decext} $\gray{\thresh{\vec B^j}{l} , }C\seqar \decext j c 1 \gray{, \thresh{\vec B^j}{l+1}}$
        \end{enumerate}
\end{lemma}
\begin{proof}
    We prove only \cref{item:decext-implies-zerocase,item:zerocase-implies-decext}, with \cref{item:decext-implies-onecase,item:onecase-implies-decext} respectively proved similarly.
    We proceed by backwards induction on $j$.

     For the base case, $j=N$, note that $\vec B^N = \epsilon$. We have:
    \begin{enumerate}
        \item By the $\dextaxs{\vec B} \BB$ axioms, we have either $\decext N c 0 \dseqar 0$ or $\decext N c 0 \dseqar A$, so we can conclude by $\id$ or $\lefrul 0$ and weakenings.
        \item 
        \begin{itemize}
            \item If $l<0$ then we have $\thresh \epsilon {l+1} \dseqar 1$ from $\thrextaxs \B\BB$, and so we conclude by $\rigrul 1 $ and weakenings.
            \item If $l>0$ then we have $\thresh \epsilon l \dseqar 0$ from $\thrextaxs \B\BB$, and so we conclude by $\lefrul 0$ and weakenings.
            \item If $l=0$ then $k=c $ by assumption, so from $\dextaxs \B\BB$ we have $\decext N k 0 \dseqar A$, whence we conclude by $\id$ and weakening steps.
        \end{itemize}
    \end{enumerate}

\newcommand{\storage}{\cref{decext-char-low-count}}
    For the inductive steps we have the following derivations:
            \[
            \eqref{item:decext-implies-zerocase} : \ 
            \vlderivation{
            \vliin{\cut}{}{\thresh {\Bj j} l , \decext j c 0 \seqar A, \thresh {\Bj j}{l+1}}{
                \vliq{\dextaxs \B \BB }{}{\thresh {\Bj j} l , \decext j c 0 \seqar A, \thresh {\Bj j}{l+1}, B_j}{
                \vliq{\thrextaxs \B \BB}{}{\thresh {\Bj j} l , \decext {j+1} c 0 \seqar A, \thresh {\Bj j}{l+1}, B_j}{
                \vliq{\IH}{}{\thresh {\Bj {j+1}} l , \decext {j+1} c 0 \seqar A, \thresh {\Bj {j+1}}{l+1}}{\vlhy{}}
                }
                }
            }{
                \vliiq{\dextaxs \B\BB}{}{B_j, \thresh {\Bj j} l , \decext j c 0 \seqar A, \thresh {\Bj j}{l+1}}{
                    \vliq{\thrextaxs \B \BB}{}{B_j, \thresh {\Bj j} l , \decext {j+1} c 0 \seqar A, \thresh {\Bj j}{l+1}}{
                    \vliq{}{}{\thresh {\Bj {j+1}} {l-1} , \decext {j+1} c 0 \seqar A, \thresh {\Bj {j+1}}{l}}{\vlhy{\text{\storage}}}
                    }
                }{
                    \vliq{\thrextaxs \B \BB}{}{B_j, \thresh {\Bj j} l , \decext {j+1} {c+1} 0 \seqar A, \thresh {\Bj j}{l+1}}{
                    \vliq{\IH}{}{\thresh {\Bj {j+1}} {l-1} , \decext {j+1} {c+1} 0 \seqar A, \thresh {\Bj {j+1}}{l}}{\vlhy{}}
                    }
                }
            }
            }
            \]
            \[
            \eqref{item:zerocase-implies-decext} : \quad 
            \vlderivation{
            \vliin{\cut}{}{\thresh{\Bj j }l , A \seqar \decext j c 0 , \thresh {\Bj j}{l+1}}{
                \vliq{\dextaxs \B\BB}{}{\thresh{\Bj j }l , A \seqar \decext j c 0 , \thresh{\Bj j}{l+1}, B_j}{
                \vliq{\thrextaxs \B\BB}{}{\thresh{\Bj j }l , A \seqar \decext {j+1} c 0 , \thresh{\Bj j}{l+1}, B_j}{
                \vliq{\IH}{}{\thresh{\Bj {j+1} }l , A \seqar \decext {j+1} c 0 , \thresh {\Bj {j+1}}{l+1}}{\vlhy{}}
                }
                }
            }{
                \vliq{\dextaxs \B\BB}{}{B_j, \thresh{\Bj j }l , A \seqar \decext j c 0 , \thresh {\Bj j}{l+1}}{
                \vliq{\thrextaxs \B\BB }{}{B_j, \thresh{\Bj {j} }{l} , A \seqar \decext {j+1} {c+1} 0 , \thresh {\Bj {j}}{l+1}}{
                \vliq{\IH}{}{\thresh{\Bj {j+1} }{l-1} , A \seqar \decext {j+1} {c+1} 0 , \thresh {\Bj {j+1}}{l}}{\vlhy{}}
                }
                }
            }
            }
            \vspace{-1.35\baselineskip} 
            \]

\end{proof}
\noindent 
Finally we consider the case when $l+c = k $ and $j\leq i$. 
The intuition here is that $B_i$ has not yet been evaluated, so the $k$-decider behaves like a bona fide decision on $B_i$, returning $A$ if false and $C$ if true.
This intuition is formalised by the following result:

\begin{lemma}
\label{decext-char-before-bool-flag}
    If $l+c = k$ then, for $j\leq i$, there are polynomial-size proofs of:
        \begin{enumerate}
            \item\label{decext-implies-dec-zerocase} $\gray{\thresh{\vec B^j}{l},} \decext jc0 \seqar A, B_i \gray{, \thresh{\vec B^j}{l+1}}$
            \item\label{dec-implies-decext-zerocase} $\gray{\thresh{\vec B^j}{l},} A \seqar B_i, \decext j c 0 \gray{, \thresh{\vec B^j}{l+1}}$
            \item\label{decext-implies-dec-onecase} $\gray{\thresh{\vec B^j}{l},} \decext j c 0 , B_i \seqar C\gray{, \thresh{\vec B^j}{l+1}}$
            \item\label{dec-implies-decext-onecase} $\gray{\thresh{\vec B^j}{l} ,} B_i , C \seqar \decext j c 0 \gray{, \thresh{\vec B^j}{l+1}}$
        \end{enumerate}
\end{lemma}

\begin{proof}
    We proceed by backwards induction on $j\leq i$.

\newcommand{\storage}{\cref{decext-char-after-bool-flag}}
\renewcommand{\storagetwo}{\cref{decext-char-low-count}}
    For the base case, $j=i$, \cref{decext-implies-dec-zerocase,dec-implies-decext-zerocase} are proved similarly:
    \[
    \eqref{decext-implies-dec-zerocase} : \quad 
    \vlderivation{
    \vliq{\dextaxs \B\BB}{}{\thresh {\Bj i}l, \decext i c 0 \seqar A, B_i, \thresh{\Bj i }{l+1}}{
    \vliq{\thrextaxs\B\BB}{}{\thresh {\Bj i}l, \decext {i+1} c 0 \seqar A, B_i, \thresh{\Bj i }{l+1}}{
    \vliq{}{}{\thresh {\Bj {i+1}}l, \decext {i+1} c 0 \seqar A, \thresh{\Bj {i+1} }{l+1}}{
    \vlhy{\text{\storage}}
    }
    }
    }
    }
    \qquad
    \eqref{dec-implies-decext-zerocase} : \quad 
    \vlderivation{
    \vliq{\dextaxs \B \BB}{}{\thresh {\Bj i}l, A \seqar   B_i, \decext i c 0, \thresh{\Bj i }{l+1}}{
    \vliq{\thrextaxs\B\BB}{}{\thresh {\Bj i}l, A \seqar  B_i,\decext {i+1} c 0,  \thresh{\Bj i }{l+1}}{
    \vliq{}{}{\thresh {\Bj {i+1}}l,  A \seqar \decext {i+1} c 0 ,  \thresh{\Bj {i+1} }{l+1}}{
    \vlhy{\text{\storage}}
    }
    }
    }
    }
    \]
    \cref{decext-implies-dec-onecase,dec-implies-decext-onecase} are proved dually, accounting for the choice when $B_i$ is true:
    \[
    \eqref{decext-implies-dec-onecase} : \quad 
    \vlderivation{
    \vliiq{\dextaxs \B\BB}{}{\thresh {\Bj i } l , \decext i c 0 , B_i \seqar C, \thresh {\Bj i }{l+1}}{
        \vliq{\thrextaxs \B\BB}{}{\thresh {\Bj i } l , \decext {i+1} c 0 , B_i \seqar C, \thresh {\Bj i }{l+1}}{
        \vliq{\wk}{}{\thresh {\Bj {i+1} } {l-1} , \decext {i+1} c 0  \seqar C, \thresh {\Bj {i+1} }{l}}{
        \vliq{}{}{\thresh {\Bj {i+1} } {l-1} , \decext {i+1} c 0  \seqar \thresh {\Bj {i+1} }{l}}{\vlhy{\text{\storagetwo}}}
        }
        }
    }{
        \vliq{\thrextaxs \B \BB}{}{\thresh {\Bj i } l , \decext {i+1} {c+1} 1 , B_i \seqar C, \thresh {\Bj i }{l+1}}{
        \vliq{}{}{\thresh {\Bj {i+1} } {l-1} , \decext {i+1} {c+1} 1 , B_i \seqar C, \thresh {\Bj {i+1} }{l}}{\vlhy{\text{\storage}}}
        }
    }
    }
    \]
    \[
    \eqref{dec-implies-decext-onecase} : \quad 
    \vlderivation{
    \vliq{\dextaxs \B\BB}{}{\thresh {\Bj i } l , B_i , C \seqar \decext i c 0 , \thresh{\Bj i }{l+1}}{
    \vliq{\thrextaxs \B\BB}{}{\thresh {\Bj {i+1} } {l-1} , C \seqar \decext {i+1} {c+1} 1 , \thresh{\Bj {i+1} }{l}}{\vlhy{\text{\storage}}}
    }
    }
    \]

    \noindent For the inductive step we have the following derivations:
        {
        \fontdimen6\textfont2=0.45em 
        \[
        \eqref{decext-implies-dec-zerocase} : \  
        \vlderivation{
        \vliin{\cut}{}{\thresh {\Bj j } l\!\!, \decext j c 0 \!\!\seqar\!\! A,\! B_i, \thresh{\Bj j }{l+1}}{
            \vliq{\dextaxs \B\BB}{}{\thresh {\Bj j } l\!\!, \decext j c 0 \!\!\seqar\!\! A,\! B_i, \thresh{\Bj j }{l+1}, B_j}{
            \vliq{\thrextaxs \B\BB}{}{\thresh {\Bj j } l\!\!, \decext {j+1} c 0 \!\!\seqar\!\! A,\! B_i, \thresh{\Bj j }{l+1}, B_j}{
            \vliq{\IH}{}{\thresh {\Bj {j+1} } l\!\!, \decext {j+1} c 0 \!\!\seqar\!\! A,\! B_i, \thresh{\Bj {j+1} }{l+1}}{\vlhy{}}
            }
            }
        }{
            \vliiq{\dextaxs \B\BB}{}{B_j, \thresh {\Bj j } l\!\!, \decext j c 0 \!\!\seqar\!\! A,\! B_i, \thresh{\Bj j }{l+1}}{
                \vliq{\thrextaxs \B\BB}{}{B_j, \thresh {\Bj j } l\!\!, \decext {j+1} c 0 \!\!\seqar\!\! A,\! B_i, \thresh{\Bj j }{l+1}}{
                \vliq{\wk}{}{ \thresh {\Bj {j+1} } {l-1}\!\!, \decext {j+1} c 0 \!\!\seqar\!\! A,\! B_i, \thresh{\Bj {j+1} }{l}}{
                \vliq{}{}{\thresh {\Bj {j+1} } {l-1}\!\!, \decext {j+1} c 0 \!\!\seqar\!\! \thresh{\Bj {j+1} }{l}}{\vlhy{\text{\storagetwo}}}
                }
                }
            }{
                \vliq{\thrextaxs \B \BB }{}{B_j, \thresh {\Bj j } l\!\!, \decext {j+1} {c+1} 0 \!\!\seqar\!\! A,\! B_i, \thresh{\Bj j }{l+1}}{
                \vliq{\IH}{}{ \thresh {\Bj {j+1} } {l-1}\!\!, \decext {j+1} {c+1} 0 \!\!\seqar\!\! A,\! B_i, \thresh{\Bj {j+1} }{l}}{\vlhy{}}
                }
            }
        }
        }
        \]
        }

        \[
\eqref{dec-implies-decext-zerocase} : \         
\vlderivation{
        \vliin{\cut}{}{\thresh{\Bj j } l , A \seqar B_i , \decext j c 0 , \thresh {\Bj j }{l+1}}{
            \vliq{\dextaxs \B\BB}{}{\thresh{\Bj j } l , A \seqar B_i , \decext j c 0 , \thresh {\Bj j }{l+1}, B_j}{
            \vliq{\thrextaxs \B\BB}{}{\thresh{\Bj j } l , A \seqar B_i , \decext {j+1} c 0 , \thresh {\Bj j }{l+1}, B_j}{
            \vliq{\IH}{}{\thresh{\Bj {j+1} } l , A \seqar B_i , \decext {j+1} c 0 , \thresh {\Bj {j+1} }{l+1}}{\vlhy{}}
            }
            }
        }{
            \vliq{\dextaxs \B\BB}{}{B_j, \thresh{\Bj j } l , A \seqar B_i , \decext j c 0 , \thresh {\Bj j }{l+1}}{
            \vliq{\thrextaxs \B \BB}{}{B_j, \thresh{\Bj {j+1} } {l-1} , A \seqar B_i , \decext {j+1} {c+1} 0 , \thresh {\Bj {j+1} }{l}}{
            \vliq{\IH}{}{\thresh{\Bj j } l , A \seqar B_i , \decext {j+1} {c+1} 0 , \thresh {\Bj j }{l+1}}{\vlhy{}}
            }
            }
        }
        }
        \]
    The derivations for \cref{dec-implies-decext-onecase,decext-implies-dec-onecase} are similar.
\end{proof}

\subsection{Putting it all together}
From here the main result of this section is an immediate consequence of our final lemma:

\begin{proof}
    [Proof of \cref{thm:imm-szel}]
    As $\dec A {B_i^k}C := \decext  0 0 0 $, the result is now just a special case of \cref{decext-char-before-bool-flag}, setting $c=0$, $l=k$ and $j=0$. 
\end{proof}

\section{Strategies to proofs: the non-deterministic case}
\label{sec:strats-to-proofs}

Our final main result is the converse of \Cref{games-simulate-proofs}, in the non-deterministic case:
\begin{theorem}
\label{thm:strat-to-proofs}
If $\bl$ has a size $N$ proof of an \exndt\ sequent $\Gamma \seqar \Delta$ over $\E$ then $\elndt$ has a $\poly (N)$ size proof of $\Gamma \seqar \Delta$ over some $\E' \supseteq \E$.
\end{theorem}
\begin{corollary}
    $\elndt$ polynomially simulates $\bl$, over \ndt\ sequents.
\end{corollary}
\noindent 
This argument is much more complicated than the deterministic case, \cref{thm:eldt-psim-db}, and requires the formalisation of Immerman-\szel\ from the previous section.
The argument follows by composing a series of intermediate polynomial simulations.








\subsection{De Morgan normal form of strategies}
\label{sec:dm-nf}
Before translating strategies to proofs, 
it will be useful to work with strategies in `De Morgan' form, where each query only has negations in front of $\elndt$ formulas, with no $\land $ in its scope.
Let us write $\blus$ for the restriction of $\bl$ to this syntax of queries.
    We can define for each $\bl$-query $Q$ an `equivalent' $\blus$-query $\dm Q$ in the natural way, by pushing negations down to \exndt\ formulas.

\begin{definition}[$\DM$-translation]\label{+translation}
    For each $\bl$-query $Q$ define a $\blus$-query $\dm Q$ by,
    \[
\begin{array}{r@{\ := \ }l}
     \dm B & B  \\ 
       \dm{(Q\land R)}& \dm Q\land \dm R \\  
        \dm{(Q\lor R)} & \dm Q\lor \dm R 
\end{array}
\qquad
\begin{array}{r@{\ := \ }l}
           \dm{(\neg B)} & \neg B \\
     \dm {(\neg \neg Q)} & \dm Q \\
       \dm {(\neg (Q\lor R))}& \dm{(\neg Q)} \land \dm {(\neg R)}\\
        \dm{(\neg (Q\land R))} &  \dm{(\neg Q)} \lor \dm{(\neg R)} 
\end{array}
\]
where $B$ is an \exndt\ formula.
\end{definition}

  We will need to lift this translation to strategies too for which, we require some intermediate results.  

Let us write $\depth Q$ for the \emphasis{depth} of $Q$, i.e.\ the maximum length of a path from the root of $Q$ to a (maximal) \exndt\ subformula.
%
Our first intermediate result 
 constructs a winning strategy from situations when the assignment is
inconsistent with some substitution of queries.

\begin{lemma}
    [Leibniz property]
    \label{replacing-queries}
    Let $P,P',Q[X]$ be $\bl$ (or $\blus$) queries, where $X$ is a distinguished single subquery occurrence in $Q[X]$.
    From $\{P\mapsto b, P'\mapsto b, Q[P] \mapsto c , Q[P'] \mapsto 1-c\}$ 
    there is a $\bl$ (resp., $\blus$) winning strategy with $O(\log(\depth {Q[X]}))$ rounds.
\end{lemma}
\begin{proof}
    Let $\pi$ be the (unique) path in the of $Q[X]$ from the root to $X$.
    We proceed by divide-and-conquer induction on the length $|\pi|$ of $\pi$.
    \begin{itemize}
        \item If $|\pi| = 0 $ then $Q[X] = X$, and $\{P \mapsto b, P'\mapsto b, P \mapsto c, P' \mapsto 1-c\}$ always contains a simple contradiction.
        \item If $|\pi|=1$ then $X$ is an immediate subquery of $Q$:
        \begin{itemize}
            \item if $Q[X] = \neg X$ then $\{P \mapsto b, P' \mapsto b, \neg P \mapsto c, \neg P' \mapsto 1-c\}$ always contains a simple contradiction;
            \item if $Q[X] = X \circ T$ for some $\circ \in \{\lor,\land\}$ and query $T$, then ask $T$.
            From here note that $\{P \mapsto b, P'\mapsto b, P \circ T \mapsto c , P' \circ T \mapsto 1-c, T \mapsto d\} $ always contains a simple contradiction. Similarly if $Q[X] = T \circ X$.
        \end{itemize}
        \item For the inductive step, write $Q[X] = Q'[R[X]]$ where $R[X]$ is the subquery of $Q[X]$ rooted (roughly) halfway along $\pi$.
        Ask $R[P] $ and $R[P'] $, with responses $d $ and $d'$ respectively:
        \begin{itemize}
            \item if $d \neq d'$ then apply the inductive hypothesis to $P,P',R[X]$;
            \item if $d=d'$ then apply the inductive hypothesis to $R[P], R[P'], Q'[X]$.\qedhere
        \end{itemize}
    \end{itemize}
\end{proof}

\noindent 
For a tree $T$ write $|T|$ for its number of leaves. Recall `Spira’s lemma’, that each binary tree has a subtree of roughly half the number of leaves:

  \begin{lemma}[Spira's lemma, \cite{spira1971time}]\label{spiral} 
  For any binary tree $T$ there is a subtree $T'$ such that $\frac{1}{3}|T| \leq |T'| < \frac{2}{3}|T|$.  
  \end{lemma}



\noindent 
We can use the above fact to drive a divide-and-conquer strategy for De Morgan duality.
Let us write $\exndtsize Q$ for the number of (maximal) \exndt\ subformula occurrences in $Q$, i.e.\ writing $Q = \phi(\vec B)$ where $\phi(\vec x)$ is a Boolean formula and $\vec B$ are maximal \exndt\ subformulas, $\exndtsize Q$ is just the number of leaves of the formula tree of $\phi(\vec x)$.

\begin{lemma}
[Duality]
\label{duality-strategy}
  Given a $\bl$-query $Q$ there is a winning strategy in $\blus$ from $\{ \dm Q \mapsto b, \dm{(\lnot Q)} \mapsto b\} $ with
  $O(\log(\exndtsize Q)$ rounds.
\end{lemma}
\begin{proof}
    We proceed by divide-and-conquer induction on $\exndtsize Q$:
    \begin{itemize}
        \item If $Q$ is just an \exndt-formula or its negation, then the statement is immediate.
        \item For the inductive step, write $Q = P[R]$ where $R$ is a subquery of $Q$ of roughly half the size, more precisely s.t.\ $\frac 1 3 \exndtsize Q  \leq \exndtsize R \leq \frac 2 3 \exndtsize Q$, by \cref{spiral}.
        Ask $\dm R$ and $\dm {(\lnot R)}$, with responses $c$ and $c'$ respectively:
        \begin{itemize}
            \item if $c = c'$ then apply the inductive hypothesis to $R$;
            \item if $c \neq c'$ then ask $\dm {(P[c])}$ and $\dm{(\lnot P[c])}$, with responses $d$ and $d'$ respectively:
            \begin{itemize}
            \item if $d\neq b$ then apply \cref{replacing-queries} to $\dm R, c, \dm {(P [X])}$;\footnote{Note that $\dm Q = \dm{(P[R])} = \dm P [ \dm R]$.}
            \item if $d'\neq b$ then apply \cref{replacing-queries} to $\dm{(\lnot R)}, c', \dm{(\lnot P[X])}$;
                \item otherwise $d=d' = b$ so apply the inductive hypothesis to $P[c]$. \qedhere
            \end{itemize}
        \end{itemize}
    \end{itemize}
\end{proof}

\noindent 
Finally we can obtain the translation of $\bl$ strategies to De Morgan form. For a $\bl$ strategy $\sigma$, write $\exndtsize \sigma := \sum\limits_{Q \in \sigma} \exndtsize Q$, where $Q\in \sigma$ means that $Q$ varies over nodes labelled by the query $Q$ in $\sigma$.

\begin{proposition}\label{game simulation}
    Given a $\bl $ strategy $\sigma$ over $\E$ from an \exndt\ sequent winning in $ d$ rounds
    there is a 
    $\blus$ strategy over $\E$ for the same sequent winning in $O(d+ \log(\exndtsize \sigma))$ rounds.
\end{proposition}

\begin{proof}
  
    From $\sigma$ construct a $\blus$ `pre-strategy' $\dm{\sigma}$
  for the same sequent 
  by replacing each query $Q$ by $\dm Q$ 
  ($\dm \sigma$ is visualised in \cref{transl-strat}).
The only paths of $\dm \sigma$ that no longer end after a simple contradiction nonetheless contain a subset of queries of form $\{\dm Q \mapsto b, \dm{(\lnot Q)} \mapsto b\}$. 
So we may extend $\dm \sigma$ into a bona fide $\blus$ winning strategy by appealing to \cref{duality-strategy}.
%
%
\end{proof}

\begin{figure}[t]
    \centering
   \begin{tikzpicture}[scale=1.32]
\begin{scope}[yscale=-1,xscale=1]
\draw (0,-0.8) -- (2,2) -- (4,-.6);

\draw (0,-0.8) .. controls (1.2,0) and (2.8,-0.6) .. (4,-.6);

\draw (2.28,1.208) -- (2.28,1);
\draw (2.24,.82) -- (2.08,.68);
\draw (2.32,.82) -- (2.48,.68);
\draw [thick,dotted](2.26,.08) -- (1.88,-.2);
\draw (2.30,.08) -- (2.6,-.2);


\node at (2,2.2) {\small };

\node at (1.44,.88) {\small };
\node at (1.68,.94) {\small  };
\node at (2.28,1.28) {\footnotesize $\neg $};
\node at (2.28,.88) {\footnotesize $\circ $};
\node at (2.068,.58) {\footnotesize $\land$};
\node at (2.48,.58) {\footnotesize $\lor$};
\node at (2.28,.38) {\footnotesize $\vdots$};

\path[draw,use Hobby shortcut,closed=true]
(1.868,.58) .. (2.068,.88) ..(2.28,1.38)..(2.48,1.08)   .. (2.68,.58) ;

\node at (2.08,-.68) {\scriptsize $\{Q\mapsto b,\neg Q\mapsto b\}$};

\end{scope}

    \end{tikzpicture}
    \begin{tikzpicture}[scale=1.32]
\begin{scope}[yscale=-1,xscale=1]
\draw (0,-0.8) -- (2,2) -- (4,-.6);
\draw (0,-0.8) .. controls (1.2,0) and (2.8,-0.6) .. (4,-.6);
\draw (2.28,1.208) -- (2.28,1);
\draw (2.24,.82) -- (2.08,.68);
\draw (2.32,.82) -- (2.48,.68);
\draw [thick,dotted](2.26,.08) -- (1.88,-.2);
\draw (2.30,.08) -- (2.6,-.2);


\node at (-.58,.88) {$\xlongrightarrow{\dm{}}$};

\node at (2,2.2) {\small };
\node at (1.58,.88) {\small  };
\node at (1.68,1) {\small  };

\node at (2.28,.92) {\footnotesize $\bar{\circ} $};
\node at (2.068,.58) {\footnotesize $\lor$};
\node at (2.48,.58) {\footnotesize $\land$};
\node at (2.28,.38) {\footnotesize $\vdots$};

\path[draw,use Hobby shortcut,closed=true]
(1.868,.58) .. (2.068,.88) ..(2.28,1.38)..(2.48,1.08)   .. (2.68,.58) ;

\node at (2.08,-.68) {\scriptsize $\{\dm{Q}\mapsto b, \dm{(\neg Q)}\mapsto b\}$};

\end{scope}

    \end{tikzpicture} 
    \caption{
Visualisation of a $\bl$ strategy (left) and its $\dm \cdot$ translation (right).
`Leaves' of the form $\{\dm Q\mapsto b, \dm {(\neg Q)}\mapsto b\}$ are not simple contradictions of $\blus$, but are rather justified by the $O(\log(\exndtsize Q))$-depth strategies from \cref{duality-strategy}.}
    \label{transl-strat}
\end{figure}

\subsection{From De Morgan strategies to $\posboolelndt$.}
We now use our formalisation of Immerman-\szel\ in order to translate De Morgan strategies to proofs without negation.
Similarly to the deterministic case, cf.~\cref{sec:strats-to-proofs-deterministic}, it will be useful to first translate to a version of $\elndt$ with formulas closed under certain Boolean combinations, this time only positively so:

\begin{definition}
[System for positive Boolean combinations of \exndt\ formulas]
Write $\posboolelndt$ for the extension of $\elndt$ that allows $\{\lor,\land\}$-combinations of \exndt\ formulas and the corresponding rules for these connectives from \cref{eq:pos-bool-rules}.
\end{definition} 

The main result of this subsection is:

\begin{proposition}
    \label{dm-strats-to-posbool-elndt}
Given a $\blus$ proof $\sigma$ over $\BB$ of an \exndt\ sequent $\Gamma \seqar \Delta $ 
there is a $\posboolelndt$ proof of $\Gamma \seqar \Delta$ over some $\BB' \supseteq \BB$ of size polynomial in $|\sigma|$.
\end{proposition}

\noindent 
Towards a proof of this result, let us fix for the remainder of this subsection $\sigma,\BB,\Gamma, \Delta$ as declared in the statement above, and let $\B = B_0, \dots, B_{N-1}$ enumerate all the \exndt\ formulas occurring in $\sigma$.
By \cref{thm:imm-szel} let $\dextaxs k i  \supseteq \thrextaxs \B \BB \supseteq \pos\BB \supseteq \BB$ be such that we have polynomial-size $\elndt$ proofs of the sequents in \cref{eq:decider-truth-conds-relativised} over $\dextaxs k i $, for $A = 1$ and $C=0$.
Using those sequents, we can simulate negated \exndt\ formulas of $\sigma$ in $\posboolelndt$ relative to a fixed number of true formulas. 
Formally, writing $\decextaxs k := \dextaxs k 0 \cup \cdots \cup \dextaxs k {N-1}$, we have:

\begin{lemma}
    \label{strategy to elndt(land)}
For each $k\in \mathbb Z$ there are $\posboolelndt$ proofs over $\decextaxs k $ of $\thresh{\List{}{B}{}}{k} , \Gamma \seqar \Delta , \thresh{\List{}{B}{}}{k+1}$ of size polynomial in $|\sigma|$. 
\end{lemma}

\begin{proof}
  
 For each query $Q$ of $\sigma$ write $Q^k$ for the result of replacing each \exndt\ subquery $\neg B_i$ by the respective decider $\decider k 1 {B_i} 0$.
 We proceed inductively on the structure
 of $\sigma$, similarly to the deterministic case, cf.~\cref{eldt-psim-booleldt}, instead conducting cuts on $Q^k$ for each query $Q$, always maintaining the invariant that if $Q\mapsto 0$ then $Q^k$ is on the RHS of a sequent, and if $Q \mapsto 1$ then $Q^k$ is on the LHS.
 We also always keep $\thresh{\List{}{B}{}}{k}$ on the LHS and $\thresh{\List{}{B}{}}{k+1}$ on the RHS.
Formally, if $\sigma $ begins by querying some $Q$, with substrategies $\sigma_0, \sigma_1$ for the answers $0,1$ respectively, we construct the following $\posboolelndt $ proof:
\[
\vlderivation{
\vliin{\cut}{}{ \thresh \B k, \Gamma \seqar \Delta, \thresh \B {k+1}}{
    \vltr{\IH}{\thresh \B k, \Gamma \seqar \Delta, Q^k , \thresh \B {k+1}}{\vlhy{\quad }}{\vlhy{}}{\vlhy{\quad}}
}{
    \vltr{\IH}{\thresh \B k, \Gamma, Q^k \seqar \Delta , \thresh \B {k+1}}{\vlhy{\quad }}{\vlhy{}}{\vlhy{\quad}}
}
}
\]
where the subproofs marked $\IH$ are obtained by the inductive hypohesis, with $\sigma_0,\sigma_1$ construed as winning strategies from the corresponding initial states.
Thus it remains to derive the base cases of simple contradictions by polynomial size proofs:
 \begin{itemize} 
     \item $\lnot$ contradictions of $\sigma$ have form either of:
     \begin{itemize}
         \item $\{B_i \mapsto 0, \lnot B_i \mapsto 0\}$: these are translated to polynomial size proofs over $\dextaxs k i $ of $\thresh{\List{}{B}{}}{k} , \Gamma \seqar B_i, $\\$ \decider k 1 {B_i} 0 , \Delta, \thresh{\List{}{B}{}}{k+1}$, by \cref{thm:imm-szel}.
         \item $\{B_i \mapsto 1, \lnot B_i \mapsto 1\}$: these are translated to polynomial size proofs over $\dextaxs k i $ of 
         $\thresh{\List{}{B}{}}{k} , \Gamma, B_i,$\\$  \decider k 1 {B_i} 0 \seqar \Delta, \thresh{\List{}{B}{}}{k+1}$, again by \cref{thm:imm-szel}.
     \end{itemize}
\item all other simple contradictions, namely the $\{\lor,\land\}$ contradictions, decision contradictions, extension contradictions and similarity contradictions, are handled in exactly the same way as for the deterministic case, in the proof of \cref{eldt-psim-booleldt}, just weakening the extraneous threshold extension variables. \qedhere
 \end{itemize}
\end{proof}

\noindent 
Now we obtain \cref{dm-strats-to-posbool-elndt}
 by composing proofs from \cref{strategy to elndt(land)}, for $k=0, \dots , N$, and appealing to the monotonicity properties of thresholds from \cref{monotonicity-of-threshold-subscripts} earlier:

    \renewcommand{\storageone}{\cref{monotonicity-of-threshold-subscripts}}
\renewcommand{\storagetwo}{\cref{strategy to elndt(land)}}
\begin{proof}
[Proof of \cref{dm-strats-to-posbool-elndt}]
    Set $\BB' := \decextaxs 0\cup \cdots \cup \decextaxs {N}$ and construct the $\posboolelndt$ proof:
    \[
    \vlderivation{
    \vliq{ \cntr}{}{\Gamma \seqar \Delta}{\vliiiiin{(N+2) \cdot \cut}{}{\Gamma, \dots, \Gamma \seqar \Delta, \dots , \Delta}{
        \vliq{}{}{\seqar  \thresh \B 0 }{\vlhy{\text{\storageone}}}
    }{
        \vliq{}{}{\thresh \B 0 , \Gamma \seqar\Delta , \thresh \B 1}{\vlhy{\text{\storagetwo}}}
    }{\vlhy{\cdots} }{
        \vliq{}{}{\thresh \B {N} , \Gamma \seqar \Delta, \thresh \B {N+1}}{\vlhy{\text{\storagetwo}}}
    }{
        \vliq{}{}{\thresh \B {N+1} \seqar}{\vlhy{\text{\storageone}}}
    }}
    }
    \vspace{-1.6\baselineskip} 
    \]
\end{proof}

\subsection{From $\posboolelndt$ to $\elndt$}
Just like in the deterministic case, we obtain a polynomial simulation of $\posboolelndt$ by $\elndt$ by the Boolean constructions of \cite[Section 5]{DBLP:conf/csl/BussDasKnop} (see also \cite{BussDasKnop19:preprint}):

\begin{proposition}
\label{posboolelndt-to-elndt}
    If $\posboolelndt$ has a size N proof over $\E$ of an \exndt\ sequent $\Gamma\seqar \Delta$ then $\elndt$ has a $\poly(N)$ size proof of $\Gamma \seqar \Delta$ over some $\E' \supseteq \E$.
\end{proposition}

In fact the above result is also obtained as a special case of the development of positive decisions in \cref{sec:pos-decs}: setting $\E' := \pos \E$, we can define $A\land B := \posdec 0 A B$, as given in \cref{pos-decision-programs}, and we can duly derive the $\lefrul \land$ and $\rigrul \land$ rules of \cref{eq:pos-bool-rules}, under this interpretation, by appealing to the truth conditions for positive decisions from \cref{truth-pos-decs}.

As promised the proof of the main result of this section is now obtained by composing the various polynomial simulations we have obtained:

\begin{proof}
    [Proof of \cref{thm:strat-to-proofs}]
    Immediate from \Cref{game simulation,dm-strats-to-posbool-elndt,posboolelndt-to-elndt}.
\end{proof}

\section{A proof complexity theoretic analogue of $\coNL=\NL$}
\label{sec:prf-comp-immszel}
In this section, we use our non-uniform formalisation of the Immerman-\szel\ result to obtain a bona fide proof complexity theoretic version of $\coNL=\NL$.

To formulate the statement we want, it is not enough to equate $\elndt$ with some `dual' version, reasoning over co-NBPs: as sequents are two-sided, having the negation of a \exndt\ formula $A$ on one side is the same as having $A$ on the other side. 
The corresponding analogue we shall formalise is rather akin to ``the Logspace Hierarchy collapses to $\NL$''.
To be concrete, a \emphasis{$\exists\forall$BP} is just an alternating branching program with only one alternation, starting with non-deterministic states (see, e.g., \cite{Weg00:bps-and-bdds} for a detailed exposition). 
We shall show in this section that $\elndt$ polynomially simulates a corresponding system for $\exists\forall$BPs.

\subsection{Systems for co-nondeterministic and $\exists\forall$-alternating branching programs}
First let us define the appropriate systems.
In terms of notation, we shall make use of explicit metavariables to delineate alternations of $\lor $ and $\land $ in alternating branching programs. 
All notation should be construed as self-contained to this section, to avoid any clashes with previous sections.
As we have already seen several examples of systems, representations and proofs in related systems, we shall remain as brief as possible in this subsection.

\newcommand{\eleadt}{\mathsf{eL}\exists\forall\mathsf{DT}}
\newcommand{\exeadt}{e$\exists\forall$DT}

A \emphasis{co-\exndt\ formula}, written $U,V$ etc., is generated by,
\[
U,V,\dots \quad ::= \quad 0 \quad \mid \quad 1 \quad \mid \quad \dec U p V \quad \mid \quad  U\land V \quad \mid \quad u_i
\]
where $u_0,u_1,\dots$ is a fresh set of \emphasis{co-\exndt\ extension variables}.

From here an \emphasis{\exeadt\ formula}, written $X,Y$ etc., is generated by,
\[
X,Y,\dots \quad ::= \quad 0 \quad \mid \quad 1 \quad \mid \quad  U\quad \mid \quad \dec X p Y \quad \mid \quad  X\lor Y \quad \mid \quad x_i
\]
where $x_0,x_1,\dots$ is a fresh set of \emphasis{\exeadt\ extension variables}.

Extension axiom sets for both types of formula above are defined as expected, i.e.\ sets of the form $\mathcal U = \{u_i \extiff U_i\}_{i<n}$ or $\mathcal X = \{x_i \extiff X_i\}_{i<n}$, with each $U_i$ or $X_i$, respectively, containing only extension variables among $u_0, \dots, u_{i-1}$ or $x_{0}, \dots, x_{i-1}$, respectively.
It is not hard to see that \exeadt\ formulas, over \exeadt\ extension axiom sets, represent $\exists\forall$BPs, just like \exndt\ formulas for NBPs.

\begin{definition}
[Systems for co-\exndt\ and \exeadt\ formulas]
 The system co-$\elndt$ is defined like  $\elndt$, cf.~\cref{eldt-and-elndt-systems} only over co-\exndt\ formulas instead of \exndt\ formulas, using $\land $-rules from \eqref{eq:pos-bool-rules} instead of the $\lor$-rules of $\elndt$.

The system $\eleadt$ is defined just like $\elndt$, cf.~\cref{eldt-and-elndt-systems}, only over \exeadt\ formulas instead of \exndt\ formulas, so including $\land $-rules from \eqref{eq:pos-bool-rules} as well as the rules of $\elndt$, over appropriate formulas.
\end{definition}

\begin{remark}
    [Arbitrary alternation]
    \label{rem:arbitrary-alternation}
    It is not hard to see from here how to develop a syntax for more general alternating branching programs, but we shall remain within the restricted systems here defined to avoid cumbersome metavariable bookkeeping. 
\end{remark}

The main result of this section is:

\begin{theorem}
[Proof complexity theoretic analogue of $\coNL = \NL$]
\label{thm:conl-nl-proofs}
    $\elndt$ polynomially simulates $\eleadt$, over NDT sequents.
\end{theorem}
\noindent 
To prove this we define a series of translations, according to different types of BPs, and combine them appropriately.
Like earlier results, we could refine this result to account for conclusions with extension variables, over corresponding sets of extension axioms. 
However, as we shall work with many systems in this section over different languages, we shall refrain from doing this to simplify theorem statements and their proofs.

\subsection{$\elndt$ polynomially simulates co-$\elndt$}
First, as a warm up result, and for later use, let us establish the expected polynomial equivalence between $\elndt$ and co-$\elndt$ over simple sequents.
We state the result in only the one direction that we shall need (the other being entirely dual):

\begin{proposition}
\label{prop:elndt-psim-coelndt}
    $\elndt$ polynomially simulates co-$\elndt$ over DT sequents.
\end{proposition}

\noindent 
The result follows by `dualising' a proof;
let us make this notion formal:

\begin{definition}
    [Negation]
    \label{def:negation}
    For each co-\exndt\ extension variable $u_j$ let its \emphasis{negation} $\bar u_j$ be a fresh \exndt\ extension variable.
    We extend negation to all co-\exndt\ formulas 
    by,
   \[
   \begin{array}{r@{\ := \ }l}
        \bar 0 & 1\\
        \bar 1 & 0 \\
        \overline{\dec U p V} & \dec {\bar U} p {\bar V} \\
        \overline{U\land V} & \bar U \lor \bar V 
   \end{array}
   \]
%
    We write $\bar \Gamma := \{ \bar U : U \in \Gamma\}$. 
    We further extend this notation to sets of co-\exndt\ extension axioms: if $\mathcal U = \{u_i \extiff U_i\}_{i<m}$ then $\bar{\mathcal U} := \{ \bar u_i \extiff \bar U_i \}_{i<m}$.
\end{definition}

Note that, for any co-\exndt\ formula $U$, its negation $\bar U$ is an \exndt\ formula.
Moreover, it is not hard to see that if $U$ is interpreted over some extension axioms $\mathcal U$, then $\bar U$ computes its negation, over $\bar{\mathcal U}$, in the expected way. 
We can make this observation proof-relevant too:

\begin{lemma}\label{lem:coelndt-to-elndt-by-duality}
    For each co-$\elndt$ proof $\pi$ over $\mathcal U$ of a co-\exndt\ sequent $\Gamma \seqar \Delta$ there is an $\elndt$ proof $\bar \pi$ over $\bar{\mathcal U}$ of $\bar \Delta \seqar \bar \Gamma$ of size polynomial in $|\pi|$.
\end{lemma}
\begin{proof}
Construct $\bar \pi$ by replacing all  co-\exndt\ formulas $U$ 
of the LHS of the sequent with $\bar U$ in the RHS and vice versa. We justify each resulting step of inference as follows:

\begin{itemize}
    \item Identity and cut steps on $U$ are translated to identity and cut steps, respectively, on $\bar U$.
    \item Extension hypotheses remain correct by definition.
    \item Each left structural step $\lefrul{\mathsf s}$ is translated to a right structural step $\rigrul{\mathsf s}$ and vice-versa.
    \item $\lefrul 0$ and $\lefrul \land$ steps are respectively translated to $\rigrul 1$ and $\rigrul \lor$ steps, and vice-versa. E.g.:
    \[
    \begin{array}{r@{\quad \leadsto\quad }l}
        \vlinf{\lefrul \land}{}{\Gamma, U \land V \seqar \Delta}{\Gamma, U,V} & \vlinf{\rigrul \lor}{}{\bar \Delta \seqar \bar \Gamma, \bar U \lor \bar V }{\bar \Delta \seqar \bar \Gamma, \bar U, \bar V}  \\
        \noalign{\medskip}
        \vliinf{\rigrul \land}{}{\Gamma \seqar \Delta, U\land V}{\Gamma \seqar \Delta, U}{\Gamma \seqar \Delta, V}
& \vliinf{\lefrul\lor}{}{\bar \Delta, \bar U \lor \bar V\seqar \bar \Gamma}{\bar \Delta , \bar U \seqar \bar \Gamma}{\bar \Delta, \bar V \seqar \bar \Gamma}
    \end{array}
    \]
    \item Finally $\lefrul p $ and $\rigrul p $ steps are translated as follows,
    \[
    \begin{array}{r@{\quad \leadsto \quad}l}
    \vliinf{\lefrul p}{}{\Gamma, \dec{U}{p}{V} \seqar \Delta}{\Gamma,U \seqar p,\Delta}{\Gamma,p,V\seqar \Delta}
&
    \vlderivation{
    \vliin{\rigrul p}{}{\bar \Delta \seqar \dec{\bar U}{p}{\bar V} ,\bar \Gamma}{
        \vliq{\circ}{}{\bar \Delta \seqar p, \bar U,\bar \Gamma}{\vlhy{\bar \Delta , \bar p \seqar  \bar U,\bar \Gamma}}
    }{
        \vliq{\bullet}{}{\bar \Delta,p\seqar \bar V,\bar \Gamma}{\vlhy{\bar \Delta\seqar \bar p, \bar V,\bar \Gamma}}
    }
    }
\\
\noalign{\medskip}
    \vliinf{\rigrul p}{}{  \Gamma \seqar \dec{  U}{p}{  V} ,  \Delta}{  \Gamma \seqar p,   U,  \Delta}{  \Gamma,p\seqar   V,  \Delta}
&
    \vlderivation{
    \vliin{\lefrul p}{}{\bar \Delta,\dec{\bar U}{p}{\bar V} \seqar \bar \Gamma}{
        \vliq{\circ}{}{\bar \Delta , \bar U\seqar p,\bar \Gamma}{\vlhy{\bar \Delta , \bar U, \bar p \seqar \bar \Gamma}}
    }{
        \vliq{\bullet}{}{\bar \Delta,p,\bar V\seqar \bar \Gamma}{\vlhy{\bar \Delta,\bar V\seqar \bar p, \bar \Gamma}}
    }
    }
\end{array}
    \]
    where the steps marked $\circ$ are obtained by cutting against constant-size proofs of $\seqar p, \bar p$, and the steps marked $\bullet$ are obtained by cutting against constant-size proofs of  $p,\bar p \seqar $. \qedhere
\end{itemize}
 \end{proof}
\noindent 
 From here, by cutting against identities on \exdt\ formulas, we immediately have a proof of our aforementioned simulation.

\begin{proof}
[Proof of \cref{prop:elndt-psim-coelndt}]
    For a \exdt\ formula $A$, it is routine to construct polynomial-size $\ldt$-proofs of $\seqar A,\bar A$ and $A,\bar A \seqar$, by induction on the structure of $A$.
    Now, given a co-$\elndt$ proof of a DT sequent $\Gamma \seqar \Delta$, first apply \cref{lem:coelndt-to-elndt-by-duality} to obtain an $\elndt$-proof of $\bar \Delta \seqar \bar \Gamma$, then apply linearly many cuts against $\seqar A,\bar A$ for each $A\in \Delta$ and against $B,\bar B \seqar $ for each $B\in \Gamma$, to obtain a proof of $\Gamma \seqar \Delta$.
\end{proof}

\subsection{Simulating $\eleadt$ relative to a fixed number of true variables}
Towards a proof of \cref{thm:conl-nl-proofs}, let us now fix, for the remainder of this section, an $\eleadt$ proof $\pi$ of an \ndt\ sequent $\Sigma \seqar \Pi$ over extension axioms $\mathcal X = \{x_j \extiff X_j\}_{j<n}$. 

\newcommand{\kk}[1]{#1^k}

Let $\vec U = U_0, \dots, U_{N-1}$ enumerate the co-\exndt\ formulas in $\pi$ that are \emph{not} \exdt\ formulas.
Write $\bar {\vec U} = \bar U_0, \dots, \bar U_{N-1}$ and $\kk U_i := \decider k 1 {\bar U_i} 0 $.
Now we can speicalise \cref{thm:imm-szel} to obtain extension axiom sets $\dextaxs k i \supseteq \thrextaxs \B \BB$, for $i<N$ and $k\in \mathbb Z$, such that: 
\begin{corollary}
\label{cor:coendt-fmla-decider-truth-conditions}
    There are polynomial size $\elndt$ proofs of,
    \begin{itemize}
        \item $\thr {\bar{\vec U}}{k}, \bar U_i, \kk {U_i} \seqar  \thr {\bar{\vec U}}{k+1}$
        \item $\thr {\bar{\vec U}}{k} \seqar \bar U_i, \kk {U_i}, \thr {\bar{\vec U}}{k+1}$
    \end{itemize}
    over $\dextaxs k i $, for $i<N$ and $k\in \mathbb Z$.
\end{corollary}
\begin{proof}
    [Proof sketch]
    From \cref{thm:imm-szel} we have polynomial size proofs of $\thr {\bar{\vec U}}{k}, \bar U_i, \kk {U_i} \seqar 0,  \thr {\bar{\vec U}}{k+1}$ and $\thr {\bar{\vec U}}{k}, 1 \seqar \bar U_i, \kk {U_i}, \thr {\bar{\vec U}}{k+1}$, over $\dextaxs k i $, whence the two sequents above follow by cutting against the initial sequents $\lefrul 0 $ and $\rigrul 1 $ respectively. 
\end{proof}

We want a sort of generalisation of this result to cover all \exeadt-theorems.
Let us first extend the definition of $\kk \cdot $ appropriately:
\begin{definition}
    [$k$-translation]
We extend the notation $\kk U_i$ to all \exeadt\ formulas of $\pi$, defining $\kk X$ for each $X$ in $\pi$ as follows,
    \[
    \begin{array}{r@{\ := \ }l}
         \kk {(XpY)} & \kk X p \kk Y \\
         \kk{(X \lor Y)} & \kk X \lor \kk Y \\
         \kk{x_j} &  e_{j}^k
    \end{array}
    \]
    where $e_{j}^k$ is a fresh \exndt\ extension variable.
We also write $\kk \Gamma  := \{\kk A : A \in \Gamma\}$, for cedents $\Gamma$ of $\pi$,
    and $\kk{\mathcal X} := \{e^k_j \extiff \kk {X_j}\}_{j<n}$.
\end{definition}

Again, let us point out that each $\kk X$ is an \exndt\ formula.
Writing $\decextaxs k := \dextaxs k 0 \cup \cdots \cup \dextaxs k {N-1}$, let us first establish the following result:

    

\begin{lemma}
\label{lem:eleadt-to-k-elndt}
    There are polynomial-size $\elndt$ proofs of $\thr {\bar{\vec U}}{k}, \Sigma \seqar \Pi, \thr {\bar{\vec U}}{k+1}$ over $\kk{\mathcal X} \cup \decextaxs k $, for each $k\in \mathbb Z$.
\end{lemma}
\begin{proof}
    We construct the required proof $\kk \pi$ from $\pi$ as follows:
\begin{itemize}
    \item Add $\thr {\bar{\vec U}}{k}$ to the LHS of each line and $\thr {\bar{\vec U}}{k+1}$ to the RHS of each line.
 \item Replace each each $U_i$ in the LHS (or RHS) by $\bar U_i$ in the RHS (respectively LHS).
    \item Replace each \exeadt\ formula $X$ that is not a co-\exndt\ formula (and so also not a \exdt\ formula) by $\kk X$.
\end{itemize}
Notice that (even extended) NDT formulas are unaffected by the formula translation described above, so 
 $\kk \pi$ certainly ends with the desired sequent, $\thr {\bar{\vec U}}{k}, \Sigma \seqar \Pi, \thr {\bar{\vec U}}{k+1}$. 
 Moreover, $\kk \pi$ is almost a correct $\elndt$ proof, it only remains to repair certain inference steps: 
\begin{itemize}
    \item Identities and other initial sequents remain derivable with additional weakenings for $\thr {\bar{\vec U}}{k}$ on the LHS and $\thr {\bar{\vec U}}{k+1}$ on the RHS.
    \item 
    Cuts remain intact by the transformation above.
    \item Any structural steps remain intact by the transformation above. 
    \item Any logical or extension step on \exdt\ formulas remains intact, as the transformation does not affect them.
    \item A LHS step on a co-\exndt\ formula $U_i$ is transformed into a RHS step on its dual $\bar U_i$, and vice versa, just like in the proof of \cref{lem:coelndt-to-elndt-by-duality}. 
    \item Logical steps on \exeadt\ formulas (that are not co-\exndt\ formulas) remain intact, as the $\kk \cdot $ translation commutes with decisions, disjunctions and extensions.
    \item The remaining critical cases are when, bottom-up, a co-\exndt\ formula `appears' for the first time, i.e.\ after a $\lefrul \lor$ or $\rigrul \lor $ step on a \exeadt\ formula when at least one auxiliary formula is a co-\exndt formula. Such steps are translated as follows,
    \[
    \begin{array}{r@{\quad \leadsto \quad}l}
             \vliinf{\lefrul \lor}{}{\Gamma, U_i \lor X \seqar \Delta}{\Gamma, U_i \seqar \Delta}{\Gamma , X \seqar \Delta}
&
    \vlderivation{
    \vliin{\lefrul \lor}{}{\thr {\bar{\vec U}}{k}, \Gamma', \kk {U_i} \lor \kk X \seqar \Delta', \thr {\bar{\vec U}}{k+1}}{
        \vliq{\bullet}{}{\thr {\bar{\vec U}}{k}, \Gamma', \kk{U_i} \seqar \Delta', \thr {\bar{\vec U}}{k+1} }{\vlhy{\thr {\bar{\vec U}}{k}, \Gamma' \seqar \bar U_i, \Delta', \thr {\bar{\vec U}}{k+1}}}
    }{
        \vlhy{\thr {\bar{\vec U}}{k}, \Gamma', \kk X \seqar \Delta}
    }
    }
\\
\noalign{\medskip}
\vlinf{\rigrul\lor}{}{\Gamma \seqar \Delta, U_i \lor X}{\Gamma \seqar \Delta, U_i , X}
&
\vlderivation{
\vlin{\rigrul \lor}{}{\thr {\bar{\vec U}}{k}, \Gamma' \seqar \Delta', \kk U_i \lor \kk X , \Delta', \thr {\bar{\vec U}}{k+1}}{
\vliq{\circ}{}{\thr {\bar{\vec U}}{k}, \Gamma' \seqar \Delta', \kk U_i, \kk X,  \thr {\bar{\vec U}}{k+1}}{
\vlhy{\thr {\bar{\vec U}}{k}, \Gamma', \bar U_i \seqar \Delta', \kk X, \thr {\bar{\vec U}}{k+1} }
}
}
}
    \end{array}    
\]
where $\Gamma',\Delta'$ are obtained from $\Gamma,\Delta $ by the formula replacement we carried out, and the steps marked $\bullet$ or $\circ$ are derived by cutting against the appropriate duality property for $\kk U_i$ from \cref{cor:coendt-fmla-decider-truth-conditions}.
In the case where $X$ is also some $U_j$, we cut against a further instance of \cref{cor:coendt-fmla-decider-truth-conditions} to convert, bottom-up, $\kk U_j$ on the LHS or RHS to $\bar U_j$ on the RHS or LHS, respectively. \qedhere
\end{itemize}
\end{proof}

\subsection{Putting it all together}
We can now assemble the proof of our main result.
Recall, from the beginning of the previous subsection, that we already fixed an $\eleadt$ proof $\pi$ of an \ndt\ sequent $\Sigma \seqar \Pi$ over extension axioms $\mathcal X = \{x_j \extiff X_j\}_{j<n}$, with co-\exndt\ formulas (that are not \exdt\ formulas) among $\vec U$.

\begin{proof}
    [Proof of \cref{thm:conl-nl-proofs}]
We have from \cref{lem:eleadt-to-k-elndt} polynomial-size $\elndt$ proofs, for each $k\in \mathbb Z$, of $\thr {\bar{\vec U}}{k}, \Sigma \seqar \Pi, \thr {\bar{\vec U}}{k+1}$, over $\kk X \cup \decextaxs k $.
We can now cut these together to obtain the $\elndt$ proof,
\[
\renewcommand{\storageone}{\ref{monotonicity-of-threshold-subscripts}}
\renewcommand{\storagetwo}{\ref{lem:eleadt-to-k-elndt}}
\vlderivation{
\vliq{\cntr}{}{\Sigma \seqar \Pi}{
\vliiiiin{(N+2)\cut}{}{\Sigma, \dots, \Sigma \seqar \Pi,  \dots, \Pi }{
    \vliq{}{}{\seqar \thr {\bar{\vec U}}{0}}{\vlhy{\text{Prop.~\storageone}}}
}{
    \vliq{}{}{\thr {\bar{\vec U}}{0}, \Sigma \seqar \Pi, \thr {\bar{\vec U}}{1}}{\vlhy{\text{Lem.~\storagetwo}}}
}{
    \vlhy{\cdots}
}{
    \vliq{}{}{\thr {\bar{\vec U}}{N}, \Sigma \seqar \Pi, \thr {\bar{\vec U}}{N+1}}{\vlhy{\text{Lem.~\storagetwo}}}
}{
    \vliq{}{}{\thr {\bar{\vec U}}{N+1} \seqar }{\vlhy{\text{Prop.~\storageone}}}
}
}
}
\]
over extension axioms $\mathcal X^0 \cup \decextaxs 0 \cdots\cup  \mathcal X^N \cup \decextaxs N$, as required.
\end{proof}

\section{Conclusions}
We proposed Prover-Adversary games $\bl$ and $\dbl$ for reasoning about non-deterministic branching programs (NBPs) and deterministic branching programs (BPs), respectively. 
We showed that $\dbl$ and $\bl$ correspond to the previously introduced systems $\eldt$ and $\elndt$, respectively, from \cite{DBLP:conf/csl/BussDasKnop,BussDasKnop19:preprint}. 
In the deterministic case, let us point out that similar ideas were communicated by Cook but, as far as we can tell, never published (see \cite{Cook01slides}).
For the non-deterministic case we formalised a non-uniform version of the Immerman-\szel\ theorem, $\coNL = \NL$ \cite{immerman1988nondeterministic,szelep1988method} to (partially) negate NBPs.

%

We applied our main technical Immerman-\szel\ construction to establish a proof complexity theoretic version of $\coNL = \NL$, namely that a system for the second level of the logspace hierarchy collapses to the first level. 
Precisely, we showed that the system $\elndt$ reasoning about non-deterministic branching programs polynomially simulates the system $\eleadt$ reasoning about alternating branching programs with two alternations, $\exists\forall$.

Let us point out that a nonuniform development of the Immerman-\szel\ result, in particular for computing negation of NBPs, has appeared before \cite{Vinay96:nbp-negation}. 
The construction in that work is different from ours as our deciders $\dec A {B_i^k} C$ only work relative to a fixed number $k$ of true inputs among $\vec B$. 
On the other hand \cite{Vinay96:nbp-negation} does not consider a \emph{formalisation} of the construction within a proof system or theory as we have done here, cf.~\cref{thm:imm-szel}.

It would be prudent to establish some \emph{bounded arithmetic} correspondence between these systems and appropriate theories of arithmetic, e.g.\ $\mathbf{VL}$ and $\mathbf{VNL}$ for $\Logspace$ and $\NL$ (see, e.g., \cite{cooknguyen2010logical}).
At the same time it would be interesting to understand better the bounded reverse mathematical strength of the Immerman-\szel\ Theorem (cf.~\cite{DBLP:conf/lics/CookK04,perron2009power}).
We are aware of ongoing work  in this direction by Beckman, Buss, Das and Knop.

We can see $\eldt$ and $\elndt$ as the canonical `Frege' systems for $\Logspace $ and $\NL$ respectively.
It would therefore also be interesting to use our games to compare (fragments of) $\elnndt$ with recently studied `OBDD proof systems' in proof complexity, e.g.\ in \cite{DBLP:conf/cp/AtseriasKV04,DBLP:conf/coco/BussIKS18,DBLP:journals/ipl/ChenZ09,itsykson_et_al:LIPIcs.MFCS.2022.59}.

\section*{Acknowledgments}
  \noindent 
  We would like to thank Arnold Beckmann, Sam Buss and Sasha Knop for several helpful discussions and related research collaborations that motivated this work.
  We would also like to thank Meena Mahajan for bringing to our attention some relevant references.

\bibliographystyle{alphaurl}
 \bibliography{book}

\end{document}